\documentclass[11pt,aps,prd,notitlepage,tightenlines,nofootinbib,showpacs,longbibliography]{revtex4-1} 
\usepackage{amsmath,amsthm,amssymb,amsfonts,url}
\usepackage{bbm}

\usepackage[usenames,dvipsnames]{color}
\usepackage{hyperref}

\hypersetup{
pdfkeywords={spinfoam, spin foam, loop quantum gravity, general relativity, coherent states, Regge calculus, semiclassical limit}, 
colorlinks=true,         
linkcolor=red,          
citecolor=red,        
urlcolor=Violet,      
}

\newcommand{\mbb}[1]{\mathbbm{#1}}
\newcommand{\ket}[1]{|#1\rangle}

\newcommand{\comment}[1]{}

\newcommand{\Mink}{\mbb R^{1,3}}
\newcommand{\SUtwoC}{SU(2)^{\mbb C}}

\newcommand{\be}{\begin{equation}}
\newcommand{\ee}{\end{equation}}

\theoremstyle{plain}
\newtheorem{definition}{Definition}
\newtheorem{theorem}{Theorem}
\newtheorem{proposition}[theorem]{Proposition}
\newtheorem{lemma}{Lemma}
\newtheorem{corollary}[theorem]{Corollary}

\begin{document}

\title{\large\bf Holonomy-flux spinfoam amplitude}
\author{Claudio Perini}
\affiliation{Centre de Physique Th\'eorique, Campus de Luminy, Case 907, F-13288 Marseille, EU}

\date{\small\today}

\begin{abstract}
We introduce a holomorphic representation for the Lorentzian EPRL spinfoam on arbitrary 2-complexes. The representation is obtained via the Ashtekar-Lewandowski-Marolf-Mour\~ao-Thiemann heat kernel coherent state transform. The new variables are classical holonomy-flux phase space variables $(h,X)\simeq \mathcal T^*SU(2)$ of Hamiltonian loop quantum gravity prescribing the holonomies of the Ashtekar connection $A=\Gamma + \gamma K$, and their conjugate gravitational fluxes. For small heat kernel `time' the spinfoam amplitude is peaked on classical space-time geometries, where at most countably many curvatures are allowed for non-zero Barbero-Immirzi parameter. We briefly comment on the possibility to use the alternative flipped classical limit.
\end{abstract}


\maketitle

\section{Introduction}
In this paper we study the Segal-Bargmann coherent state transform of a local $SU(2)$ holonomy formulation of the EPRL Lorentzian spinfoam model \cite{Engle:2007wy}, extending a previous analysis \cite{Bianchi:2010mw} to an arbitrary number of vertices. The holomorphic transform defines a tentative coherent state path integral for loop quantum gravity. The $SU(2)$ holonomy amplitude as well as the holomorphic amplitude are obtained as the composition of local amplitudes with canonical boundary Hilbert spaces associated to each vertex. The local amplitudes are the holomorphic vertex amplitude and the anti-holomorphic face amplitude. The role of the face amplitude is to glue the vertex amplitudes together to form the full amplitude for the 2-complex. The formalism provides a close contact of the spinfoam covariant amplitudes with the Hamiltonian loop quantum gravity framework. 

The aim of this paper is to test the dynamics provided by the spinfoam amplitude using coherent states of loop quantum gravity peaked on large areas. Following Brian Hall \cite{Hall1}, the labels of coherent states are elements in the complexification of $SU(2)$. For this reason, the passage from the $SU(2)$ holonomy amplitude to the coherent state path integral is quite immediate. The Hall coherent states had many applications in Hamiltonian loop quantum gravity \cite{Thiemann:2000bx,Sahlmann:2002qj,Sahlmann:2002qk,Bahr:2006hm}, starting from a seminal paper of Ashtekar, Lewandowski, Marolf, Mour\~ao and Thiemann \cite{Ashtekar:1994nx}, as well as in the covariant spinfoam formalism \cite{Bianchi:2009ky,Bianchi:2010mw,Bianchi:2010zs,Magliaro:2010qz,Bianchi:2011ym,Roken:2010vp}.

The complexified labels $H\in\SUtwoC$ mark a point in the phase space of general relativity with Ashtekar variables, in the following way: any $H$ can be uniquely mapped into a holonomy-flux canonical pair $(h,X)$. For this reason we shall also use the name holonomy-flux representation.

We review and further develop a previous analysis of the holomorphic vertex amplitude \cite{Bianchi:2010mw}. In the semiclassical limit we find the equation
\begin{align}\nonumber
A=\Gamma+\gamma K,
\end{align}
relating the Ashtekar connection to the intrinsic ($\Gamma$) and extrinsic ($K$) curvature of a constant-time hyper-surface, where $\gamma$ is the real Barbero-Immirzi parameter. This is our first result. 

Moreover, we extend our study to a 2-complex with an arbitrary number of vertices. We analyze the anti-holomorphic face amplitude and find that its role is to constrain the set of possible space-time curvatures in the semiclassical limit. The curvature constraint turns out to be very strong in the limit considered here, which is large areas. It turns out that the allowed scalar curvatures $\Theta$ at each face of the foam have the form
\begin{align}\nonumber
\gamma\Theta=0\mod 4\pi,
\end{align}
which is a non-trivial constraint provided that $\gamma$ is not zero. The result is conditional in that we consider large areas \emph{before} taking the integral on the coherent state labels. More precisely, we consider the spinfoam holomorphic partial amplitude, namely the amplitude for a \emph{fixed} decoration of the 2-complex with coherent state labels, and study the peakedness properties of the amplitude for large fluxes, namely for large areas.

We do not know at this stage whether the results of this paper extend automatically to the full amplitude. A more detailed analysis of this subtle point will be given elsewhere.

It was known \cite{Mamone:2009pw,Bonzom:2009hw} that the current spinfoam models with simplicity constraints might suffer from a flatness issue. Instead of yielding Ricci-flat geometries in the classical limit, as expected from the quantum gravity path integral without matter, they might yield completely flat geometries. Our analysis of the holomorphic Lorentzian amplitude at large areas provides a result which is compatible with the recent microlocal analysis of the Euclidean model \cite{Hellmann:2012kz}. We find that the holomorphic partial amplitude is peaked on space-time geometries which are flat, or they possess an accidental curvature in a countable set.

Interestingly, our analysis brakes down in the limit of small Barbero-Immirzi parameter. For vanishing $\gamma$ the full continuous set of curvatures is restored. This is in agreement with a previous argument \cite{Magliaro:2011dz,Magliaro:2011zz} in the flipped semiclassical regime.

The tools used for the proofs are a cocktail of the following techniques: Nottingham asymptotic techniques \cite{Barrett:2009mw}, graviton propagator techniques \cite{Bianchi:2009ri, Rovelli:2011kf}, and the semiclassical analysis of Hall coherent states \cite{Bianchi:2009ky,Bianchi:2010mw}.

The paper is organized as follows. In the next section we review the local $SU(2)$ holonomy formulation of the Lorentzian EPRL model in terms of holonomy vertex and face amplitudes. In section \ref{section: holonomy-flux} we discuss the loop quantum gravity coherent states based on the Brian Hall proposal, their geometric content and in section \ref{section: semiclassicality} their semiclassical properties. In section \ref{section: holomorphic path integral} we build a loop quantum gravity coherent state path integral on a 2-complex via the Segal-Bargman coherent state transform. The transform is performed with respect to the previously introduced coherent states. In sections \ref{section: semiclassical vertex} and \ref{section: semiclassical face} we analyze the peakedness properties of the holomorphic vertex and face amplitudes respectively, in the semiclassical limit of small heat kernel `time', which corresponds to large areas with small relative dispersions. From section \ref{section: 4-simplex} through \ref{section: curvature constraint} we specialize the general results to a simplicial 2-complex: the semiclassical peaks of the holomorphic partial amplitude correspond to Regge-like geometries with a strong constraint on the deficit angles. Conclusions and outlooks are in the last section. The proofs are reported in the appendix.
\section{SU(2) holonomy amplitude}
A spinfoam quantum gravity amplitude \cite{Perez:2012wv} can be defined on a truncation of the theory determined by a 2-complex $\sigma$, a mathematical model for the space-time foam. The continuum limit is expected to be recovered by taking the infinite refinement limit of the 2-complex, similarly to what is done in lattice gauge theories, or by summing over 2-complexes \cite{Rovelli:2010qx}, e.g. using a group field theory vertex expansion \cite{Oriti:2009wn}.

In this paper we consider only the truncated amplitudes as we do not address the issue of the continuum limit.

Let us start with some definitions. The spinfoam 2-complex $\sigma$ is the 2-skeleton of a 4-dimensional dual complex $\mathcal C^*$. Thus $\sigma$ is the set of $n$-cells of $\mathcal C^*$, with $n=0,1,2$. The complex $\mathcal C$ and its dual $\mathcal C^*$ should be thought as a space-time cellular decomposition\footnote{More precisely, we can take $\mathcal C$ to be a CW-complex. Another possibility is to keep only the combinatorial structure and work with combinatorial complexes.}.

Thus the basic objects of the 2-complex are the 0-dimensional vertices $v$ (dual to the 4-cells $v^*$ of $\mathcal C$), the 1-dimensional edges $e$ (dual to the 3-cells $e^*$ of $\mathcal C$), and the 2-dimensional faces $f$ (dual to the 2-cells $f^*$ of $\mathcal C$). It is useful to partition each face in wedges: a wedge labeled with the couple $vf$ is a portion of the face $f$ that includes the vertex $v$. In $\mathcal C$ as well as in $\mathcal C^*$ there is a boundary map $\partial$ that defines the boundary $(n-1)$-cells of the $n$-cells.

The 2-complex alone carries no metric information. However the decorated 2-complex carries a labeling of the states of quantum geometry. The spinfoam amplitude for a 2-complex $\sigma$ is defined as the sum over all admissible decorations of $\sigma$, with specific weights for the various components of the 2-complex, and defines a tentative quantum gravity path integral. The decoration, also called coloring, can be done using various equivalent set of variables. The most common found in the literature is the spin-intertwiner coloring, but other equivalent variables revealed to be useful, such as holonomies \cite{Magliaro:2010ih,Bahr:2012qj} and fluxes \cite{Baratin:2010nn}. Here we consider $SU(2)$ holonomies and their conjugate $su(2)$ fluxes.

We provide also an orientation to the 2-complex, which is very important to write the amplitudes correctly, even though the spinfoam model considered here is orientation-independent. For our purposes, it is sufficient to assign an orientation only to the faces of the foam. The orientation of a face $f$ induces an orientation on each edge in $\partial f$. Thus we can define the source $s(ef)$ and target $t(ef)$ vertices of an edge $e$ with respect to a face bounded by $e$. Also, we can define the source $s(vf)$ and target $t(vf)$ edges of a vertex $v$ with respect to a face containing $v$, meaning that the source edge is the one on which the face induces an orientation in-going at the vertex. Let us also define a sign $\epsilon_{vef}$ which is $1$ if the edge $e$ is source with respect to $v$ and $f$, $-1$ otherwise.

A 4-cell $v^*$ of $\mathcal C$ has a boundary $\partial v^*$, which is a 3-complex. The 1-skeleton of the dual boundary $(\partial v^*)^*$ is a graph that we denote $\Gamma_v$, to which a loop quantum gravity kinematical Hilbert space $\mathcal H_{\Gamma_v}$ of quantum 3-geometry is associated:
\begin{align}\label{kinematics}
\mathcal H_{\Gamma_v}=L^2(SU(2)^{L}/SU(2)^N),
\end{align}
for a graph with $L$ links and $N$ nodes. We can think of it as the space of gauge-invariant square-integrable functions of $L$ $SU(2)$ variables, with inner product given by the Haar measure. Gauge-invariance at the nodes provides the solution of the quantum Gauss constraint\footnote{The kinematical Hilbert space solves also the spatial-diffeomorphism constraint, suitably formulated, if instead of working with embedded graphs we work with equivalence classes of embedded graphs, under the transformations induced by the complex automorphisms.}. The Hilbert space $\mathcal H_{\Gamma_v}$ provides the kinematics of the theory, locally at each vertex.

The truncated dynamics is provided by the spinfoam amplitude, a sort of partition function associating a complex number to the 2-complex $\sigma$. One of the building blocks of the spinfoam amplitude is the vertex amplitude. 

The vertex amplitude is a generalized function of the $SU(2)$ holonomies along the links of this graph. The links of $\Gamma_v$ can be labeled with the faces $f$ that are bounded by $v$, thus with a couple vertex-face $vf$. Notice that the links and the wedges share the labels $vf$. Now we are ready to define the Lorentzian EPRL vertex amplitude in the $SU(2)$ holonomy representation,
\begin{align}\label{holonomy vertex}
W(h_{vf})=\int_{SL(2,\mbb C)} \prod_{e\supset v}dG_{ve}\prod_{f\supset v}P(G_{t(vf)v}G_{vs(vf)},h_{vf}),
\end{align}
where the integrals implement the local Lorentz invariance at the vertex. Indeed we recall that $SL(2,\mbb C)$ is the double cover of the proper orthochronous Lorentz group $SO^+(1,3)$. Notice that we have one Lorentz transformation for each edge that is bounded by the vertex, or equivalently for each node in the vertex graph $\Gamma_v$. This expression is formal. However, for a large class of graphs that we call EPRL-integrable graphs, the vertex amplitude is well-defined\footnote{The holonomy vertex amplitude is well-defined if it is well-defined as a generalized function.} once we drop one redundant  $SL(2,\mbb C)$ integration  \cite{Engle:2008ev, Kaminski:2010qb}. The resulting regularized amplitude is independent of the group element on which we are not integrating. Thus we shall work implicitly only with 2-complexes $\sigma$ such that all vertex graphs are EPRL-integrable, and the redundant integrations are dropped. This regularization is understood throughout the paper.

The fundamental amplitude used to build the vertex amplitude \eqref{holonomy vertex} is the wedge amplitude. This is the integral kernel \cite{Bianchi:2010mw} of the map that implements the simplicity constraints, defined as
\begin{align}\label{EPRL kernel}
P(G,h):=\sum_j (2j+1)\text{Tr}[Y^\dagger D^{\gamma j,j}(G^{-1})YD^{j}(h)],
\end{align}
where $D^j(h)$ is the spin-$j$ $SU(2)$ representation operator, and $D^{\gamma j,j}$ the $(\gamma j,j)$ $SL(2,\mbb C)$ representation operator of the principal series. In $SL(2,\mbb C)$ there is a $SU(2)$ subgroup that leaves a reference time-like vector $\mathcal N:=(1,0,0,0)$ invariant, and the irreducible representations of $SL(2,\mbb C)$ decompose into an orthogonal sum of irreducible representation of this $SU(2)$ subgroup, labeled by a spin $k$,
\begin{align}\label{tower}
\mathcal H^{SL(2,\mbb C)}_{\gamma j,j}=\bigoplus_{k\geq j}\mathcal H^{SU(2)}_k.
\end{align}
The map $Y$ is the isometric injection of the spin-$j$ $SU(2)$ irreducible into the lowest $k=j$ $SU(2)\subset SL(2,\mbb C)$ irreducible in the tower \eqref{tower}. This completes the definition of the holonomy vertex amplitude.

The elementary wedge amplitude \eqref{EPRL kernel} plays a role in the dynamics. Indeed if we think a foam as a space-time cellular decomposition, or triangulation in some cases, space-time curvature is obtained as a sum of wedge extrinsic curvatures.

The full spinfoam amplitude is obtained once we prescribe a way to `glue' the vertex amplitudes together. The glue is the face amplitude. The standard choice of face amplitude is a $SU(2)$ Dirac delta function evaluated at the ordered product of the variables $h_{vf}$ looping around a face. This yields the full amplitude
\begin{align} \label{holonomy full}
Z=\int_{SU(2)} \prod_{vf}dh_{vf}\prod_{v}W(h_{vf})\prod_{f}\delta(\prod_{v\subset f} h_{vf}),
\end{align}
for a 2-complex without boundary. This formula is easily generalized to a 2-complex with boundary, where the amplitude is a function of the boundary holonomies, while the bulk holonomies are integrated over.

The product inside the delta function is ordered according to the orientation of the face. Notice that the arbitrary choice of the first element in the ordered product of the face amplitude is irrelevant, due to the properties of the Dirac delta function.

This is the expression of the spinfoam partition function derived in \cite{Bianchi:2010mw}, and studied in great detail in \cite{Bianchi:2010fj,Magliaro:2010ih}. We can see that the decoration of the 2-complex is simply given by $SU(2)$ holonomies on the wedges. For each wedge $vf$ we have two copies of $h_{vf}$, one in the vertex amplitude, and one in the face amplitude.

The `sum' over all possible holonomies has to be interpreted as a proposal for the path integral quantization of general relativity in the Ashtekar connection formulation, thus as a covariant formulation of loop quantum gravity. We recall that for some 2-complexes the amplitude \eqref{holonomy full} may still contain divergencies due to the $SU(2)$ integrals of the distributions \cite{Perini:2008pd,Krajewski:2010yq}.

Finally, we stress that the $SU(2)$ holonomy partition function is completely equivalent to the spin-intertwiner partition function by which the model was originally defined, with the choice of the spin-$j_f$ Hilbert space dimension, $2j_f+1$, for the face weights.
\section{Holonomy-flux observables and coherent states}\label{section: holonomy-flux}
It is not easy to extract the semiclassical behavior of the model from the holonomy amplitude \eqref{holonomy full}, for it is written in the sole configuration variables, and the conjugate momenta are left undetermined. The possibility to use coherent states optimally localized in phase space as a tool to test the dynamics of the theory has been often advocated in spinfoam quantum gravity. The novelty of the present approach, which is a development of a previous work \cite{Bianchi:2010mw}, is that we work directly with the classical Ashtekar variables of general relativity.

The canonical phase space of general relativity is infinite-dimensional and associated to fields on a 3-dimensional space-like surface, but the spinfoam truncation to a finite 2-complex $\sigma$ induces a truncation of the phase space down to a finite number of degrees of freedom. The degrees of freedom of the truncated phase space live on a graph, the 1-skeleton of a 3-complex. Consider a graph $\Gamma$ with $L$ oriented links and $N$ nodes. The classical phase space of loop gravity truncated to the graph is given by
\begin{align}\label{lqg phase space}
\mathcal T^*SU(2)^L//SU(2)^N,
\end{align}
the cotangent bundle of $L$ copies of $SU(2)$, modulo gauge transformations at the nodes of the graph. The double quotient $//$ is the symplectic reduction with respect to the gauge transformations that act at the nodes of the graph. Notice that this phase space is identical to the one of a lattice $SU(2)$ gauge theory, and we stress that this is a truncation of the classical theory which has nothing to do with the quantum theory. A nice study of this truncated classical phase space and its relation to the continuum theory can be found in the recent literature \cite{Freidel:2010aq, Freidel:2011ue, Haggard:2012pm}.

The interpretation of \eqref{lqg phase space} is the following: the $SU(2)$ transformations associated to the links are the holonomies $h_l$ of the Ashtekar connection $A$ along the links of the graph. Their conjugate $su(2)$ Lie algebra variables are the fluxes $X_l$ of the gravitational `electric' field $E$ across 2-surfaces $l^*$ which intersect once with the links, thus dual to the links. We recall that all the quantities $X_l$ should be thought as source fluxes, since they are transported to the source node $s(l)$ of the link. The target fluxes are defined by the parallel transport matrix of the Ashtekar connection, namely as $h_l\triangleright X_l$. 

The phase space structure is given by the well-known holonomy-flux Poisson algebra of general relativity,
\begin{align}
&\{h_l,h_{l'}\}=0,\,\, \quad \{X^i_l,X^j_{l'}\}=\delta_{ll'} \epsilon^{ij}_{\;\;\,k} X^k_l,\nonumber\\\label{poissonsmeared}
&\{X^i_{l},h_{l'}\}=\pm\delta_{ll'}8\pi G\gamma \,\tau^i h_l,
\end{align}
where the sign specifies the relative orientation of the link with respect to its dual surface. This is often called a smeared algebra, for it can be derived `integrating' the canonical brackets
\begin{align}
&\{A_a^i(x),A_b^j(y)\}=0,\,\, \quad \{E^a_i(x),E^b_j(y)\}=0, \nonumber\\
&\{E^a_i(x),A_b^j(y)\}=8\pi G\gamma \,\delta^i_j\delta_a^b\delta(x,y).
\end{align}
For an interpretation of the phase space \eqref{lqg phase space} in terms of a collection of polyhedra see \cite{Bianchi:2010gc}.

In the quantum theory, a coherent state is optimally localized in phase space and thus it is labeled by a set of holonomy-flux pairs
\begin{align}
(h_l,X_l)\in\mathcal T^*SU(2)\simeq SU(2)\times su(2),
\end{align}
one for each link $l\subset\Gamma$ of the graph. Notice that the previous labels specify a point in phase space only up to gauge transformations, so that they have a certain degree of redundancy.

A loop quantum gravity coherent state for the graph $\Gamma$, localized at a phase space point $H_{l}\in SU(2)^\mbb C$ ($l\subset \Gamma$) is defined as
\begin{align} \label{coherent state}
\Psi^t_{H_{l}}(h_{l}):=\int_{SU(2)} \prod_{n}{dg_n}\prod_{l} K_t(g_{t(l)}h_{l}g^{-1}_{s(l)},H_{l}),
\end{align}
where the labels $H_{l}\in SU^{\mbb C}$ belong to the complexification of $SU(2)$, which is $SL(2,\mbb C)$ viewed as a complex manifold, and $K_t$ is the analytic continuation in the second argument of the heat kernel $K_t(h,h')$ over $SU(2)$. To fix the ambiguities, we recall that our definition of the heat kernel in terms of $SU(2)$ irreducible characters (traces) is the following,
\begin{align}
K_t(h,h'):=\sum_j(2j+1)e^{-tj(j+1)}\chi^{(j)}(h^{-1}h').
\end{align}
In formula \eqref{coherent state}, there is one $SU(2)$ integral per each node $n$ of the graph (this is the group averaging on the gauge transformations) and one heat kernel per each link $l$ of the graph. The heat kernel parameter $t$, the `time' of a fictitious diffusion process, is a semiclassicality parameter. Thus the labels of coherent states are complexified Ashtekar holonomies.

In a few words, and in the language of the well-known quantum-mechanical Gaussian wave-packets, the $SU(2)$ heat kernel $K_t(h,h')$ is a natural group `Gaussian' in the variable $h$, peaked on $h=h'$. The complexification of $h'$ serves to add a `phase' that peaks the wave-function on the desired conjugate variable. Finally, the $SU(2)$ integrals in \eqref{coherent state} project on the gauge-invariant subspace of interest by group-averaging.

We shell use the notation $SU(2)^\mbb C$ for the space of coherent state labels, even if as a manifold it is isomorphic to $SL(2,\mbb C)$. This will eliminate any confusion between the phase space labels and the $SL(2,\mbb C)$ variables that implement the local Lorentz invariance of the model.

The geometric content of the complexified labels, together with their relation to the phase space of general relativity, is easily recovered using the polar decomposition of $SL(2,\mbb C)$ in rotations and boosts,
\begin{align} \label{polar}
H=he^{iX/t^\beta},
\end{align}
where $h$ is a $SU(2)$ element and
\begin{align} \label{XwithPauliMat}
X=|X|\hat X\cdot \vec\tau
\end{align}
belongs to $su(2)$, where $\vec\tau:=-i\vec\sigma/2$ are the standard $su(2)$ generators, thus the $i$ in the exponent turns \eqref{XwithPauliMat} into a boost generator in $sl(2,\mbb C)$. The quantity $\hat X$ denotes a normalized vector in $\mbb R^3$. Notice also that we have the $SU(2)$-invariant inner product $X\cdot Y:=\text{tr}(X^\dagger Y)$ defined on the Lie algebra $su(2)$, that we used implicitly to define the norm $|X|$ in the formula \eqref{XwithPauliMat}. 

Differently from \eqref{poissonsmeared}, we chose $X$ in the polar decomposition to be dimensionless and proportional to the gravitational flux. The precise relation between $X$ and the flux is discussed in section \ref{section: semiclassicality}.  From now on we will stay with this choice. Moreover, without loss of generality we have introduced in \eqref{polar} a power law scaling factor $1/t^\beta$ of the flux variable $X$, with $\beta$ a positive real number. This is useful in order to study a class of semiclassical states at once. At the end of section \ref{section: semiclassicality} we discuss the semiclassical properties of coherent states with such dependence on the heat kernel time.

Through \eqref{polar} we recover in a unique way the geometric holonomies and gravitational fluxes from the coherent state $\SUtwoC$ labels. Notice that we have slightly changed the notation since

More than this, we have the diffeomorphism
\begin{align}
\SUtwoC\simeq \mathcal T^*SU(2)\simeq SU(2)\times su(2),
\end{align}
which defines in fact a natural symplectomorphism. Indeed we recall \cite{Hall1} that the complex structure of $\SUtwoC$ and the phase space structure of $\mathcal T^*SU(2)$ fit together so as to form a K\"ahler manifold, so that there is a natural way to think $\SUtwoC$ as the phase space of $SU(2)$.

The important property of the coherent states that allows us to build a coherent state path integral for quantum gravity is the fact that they span the loop quantum gravity kinematical Hilbert space. Indeed in the Hilbert space $L^2(SU(2))$, the one associated to a graph which is a single loop, we have the following resolution of the identity
\begin{align} \label{resolution loop}
\int_{\SUtwoC} d\Omega(H)\Psi^t_{H}(h)\overline{\Psi^t_{H}(h')}=\delta(h,h'),
\end{align}
which is easily generalized to an arbitrary graph. For the explicit expression, and an elementary derivation of the $t$-dependent measure $d\Omega$, see \cite{Bianchi:2009ky}. Notice that for the single-loop Hilbert space the coherent states are simply $\Psi_H^t(h):=K_t(h,H)$. More generally, the resolution of the identity for the Hilbert space associated to a graph $\Gamma$ reads
\begin{align}
\int_{\SUtwoC} \big(\prod_{l}d\Omega(H_{l})\big)\Psi^t_{H_{l}}(h_{l})\overline{\Psi^t_{H_{l}}(h'_{l})}=\delta_{\text{g}}(h_{l},h'_{l}),
\end{align}
where the r.h.s. is a gauge-invariant delta function.
\section{The semiclassicality parameter $t$}\label{section: semiclassicality}
We review some well-know results \cite{Thiemann:2000ca,Bianchi:2009ky} on the peakedness properties of the loop quantum gravity coherent states, in particular the computation of the expectation values of some geometric operators on coherent states, and their dispersions in the semiclassical limit. This will confirm the geometric interpretation, as stated in section \ref{section: holonomy-flux}, of the $\SUtwoC$ labels we are using. Moreover, it will clarify the regime of the dynamics tested by these states.

We discuss for simplicity the simplest graph, which is a single loop. In this case the kinematical Hilbert space is simply $\mathcal H=L^2(SU(2))$, and we have a single coherent state label $H\simeq(h,X)$. Let us first define the two quantities\footnote{Notice that $j_0$ is by definition a continuous variable.}
\begin{align}\label{j0 and sigma0}
2j_0+1:=\frac{|X|}{t^{\beta+1}},\quad \sigma_0:=\frac{1}{2t},
\end{align}
with $\beta>0$. A complete basis of $SU(2)$ functions for the Hilbert space $\mathcal H$ is given by the spin-$j$ characters $\chi^{(j)}$. The area operator of a 2-surface punctured by the loop acts on basis vectors as
\begin{equation}
\hat{A} \,\chi^{(j)}(h)= 8\pi L^2_P\gamma \sqrt{j(j+1)}\,\chi^{(j)}(h).
\label{eq:A2 chi}
\end{equation}
In the limit of small $t$ the expectation value of the area operator on a coherent state is easily computed,
\begin{equation}
\langle A\rangle=\frac{(\Psi_{H},\,\hat{A}\,\Psi_{H})}{(\Psi_{H},\Psi_{H})}\approx8\pi L^2_P\gamma\sqrt{j_0(j_0+1)},
\end{equation}
which thanks to \eqref{j0 and sigma0} confirms the interpretation of $|X|$ as the quantity that prescribes the expectation value of the area, as well as the interpretation of $t$ as a semiclassicality parameter. Notice that with our parametrization \eqref{polar} of coherent states we have to remember that the dimensionless $X$ is not the gravitational flux, but it is related to it by the proportionality constant $8\pi G\gamma/t^{\beta+1}$.

Now we consider the other fundamental observable acting on the Hilbert space $\mathcal H$: the Wilson loop operator. This is the holonomy of the Ashtekar connection for the loop. Recall that it acts on basis vectors as
\begin{equation}
\hat{W} \,\chi^{(j)}(h)= \chi^{(\frac{1}{2})}(h)\,\chi^{(j)}(h)=\chi^{(j+\frac{1}{2})}(h)+\chi^{(j-\frac{1}{2})}(h).
\label{eq:W chi}
\end{equation}
As a result, for the expectation value on a coherent state we find
\begin{equation}
\langle W \rangle= 2\cos(\xi/2)\,e^{-\frac{t}{8}},
\end{equation}
where the angle $\xi$ identifies the conjugacy class of the $SU(2)$ group element where the Ashtekar loop holonomy is peaked on. Similarly, we can compute the dispersions of the area operator and of the Wilson loop. We find
\begin{equation}\label{DeltaA}
\Delta A=\sqrt{\langle A^2\rangle-\langle A\rangle^2}=4\pi L^2_P\gamma\,\sqrt{2\sigma_0},
\end{equation}
and
\begin{equation}\label{DeltaW}
\Delta W=\sqrt{\langle W^2\rangle-\langle W\rangle^2}=\sin (\xi/2)\,\frac{1}{\sqrt{2\sigma_0}}.
\end{equation}
As the area and the Wilson loop are non-commuting operators, we cannot make both their dispersions vanish at the same time. Small heat kernel time means that the state is sharply peaked on the holonomy, while large heat kernel time means that the state is sharply peaked on the area. A good requirement of semiclassicality is that the relative dispersions of both operators vanish in the limit $t\rightarrow 0$. This requirement is satisfied by the coherent states under consideration using the parametrization \eqref{polar}. Indeed using \eqref{j0 and sigma0} in \eqref{DeltaA} and \eqref{DeltaW} we find the following behavior in $t$ for the relative dispersions,
\begin{equation}\label{relative small}
\frac{\Delta A}{\langle A\rangle}\sim t^{\beta}\sqrt{t}\;\qquad \text{and}\qquad \frac{\Delta W}{\langle W\rangle}\sim \sqrt{t}.
\end{equation}
The single loop analysis can be easily generalized to an arbitrary graph. Thus we have that the loop quantum gravity coherent states for a graph $\Gamma$ are optimally peaked on the fundamental non-commuting loop quantum gravity observables. The small $t$ regime tested by these states is equivalent to looking at large areas with small relative dispersions.

Finally, we observe that the strictly positive case $\beta>0$ considered here is the same of reference \cite{Bianchi:2009ky}. However, there is a limiting case $\beta=0$ not discussed here which corresponds to relative dispersions \eqref{relative small} that vanish symmetrically and coherent state labels constant in $t$. In fact this is the case most studied by Thiemann and collaborators. The detailed treatment of this case is in progress and will be reported in a separate work. 
\section{Holomorphic path integral}\label{section: holomorphic path integral}
The Segal-Bargmann coherent transform with respect to the previously introduced coherent states allows us to rewrite the spinfoam amplitude as a state sum over coherent states. The transform is defined via the inner product between a coherent state and a general state. For functions of one $SU(2)$ variable, the transform is a map
$$\mathcal \rho:L^2(SU(2))\rightarrow \mathcal HL^2(\SUtwoC),$$
that takes a function $f$ of $SU(2)$ to the holomorphic function $\rho f$ of $\SUtwoC$ defined as
\begin{align}
\mathcal \rho f(H):=\int_{SU(2)}dh f(h)\Psi^t_H(h).
\end{align}
For functions of many $SU(2)$ variables, the transform is simply the multiple transform with respect each of them. For the particular case of gauge-invariant functions, the multiple transform is equivalently defined via the inner product with a loop quantum gravity gauge-invariant coherent state \eqref{coherent state}. In the canonical context, the coherent transform was first studied in \cite{Ashtekar:1994nx}.

We associate a phase space to the variables on the 3-dimensional surfaces of the complex $\mathcal C$. Let us see how the complex $\mathcal C$ determines a truncation in the loop quantum gravity phase space.  One possibility is the following. As discussed previously, a vertex $v$ is dual to the 4-cell $v^*$, which has a 3-dimensional boundary. We associate to this boundary a phase space defined as the loop quantum gravity phase space for the graph $\Gamma_v$. Thus we associate a holonomy-flux couple $(h_{vf},X_{vf})$ to each link $vf$ in $\Gamma_v$. The holonomy is along the link $vf$, and the flux $X_{vf}$ is across the 2-cell of $\partial v^*$ which is dual to the link $vf$. The holonomy-flux algebra of observables associated to the graph $\Gamma_v$ defines the truncated phase space associated to each vertex $v$ of the 2-complex.

Now we use the Segal-Bargmann transform. In particular, the generalization to compact groups obtained by Hall, which was adopted in the loop quantum gravity framework long time ago \cite{Ashtekar:1994nx}.

Following \cite{Bianchi:2010mw}, our first step is the construction of a holomorphic vertex amplitude. This is defined as the Segal-Bargmann transform of the holonomy vertex amplitude, namely
\begin{align} \label{holomorphic vertex}
W_t(H_{vf}):=\int_{SU(2)}\prod_{f\supset v}dh_{vf}W(h_{vf})\Psi^t_{H_{vf}}(h_{vf}).  
\end{align}
Notice that this is a slight generalization of the original transform to generalized functions. Indeed remember that the holonomy vertex amplitude is a distribution. However its Segal-Bargmann transform is an ordinary function, due to the heat kernel dumping factor. Introducing the holomorphic wedge amplitude, or holomorphic simplicity constraint kernel,
\begin{align}\label{holomorphic wedge}
P_t(G,H):=\int_{SU(2)}dhP(G,h)\Psi^t_H(h),
\end{align}
namely the coherent transform of the simplicity constraint kernel \eqref{EPRL kernel}, we can rewrite the holomorphic vertex amplitude in the following way,
\begin{align}\label{holomorphic vertex explicit}
W_t(H_{vf})=\int_{SL(2,\mbb C)} \prod_{e\supset v}dG_{ve}\prod_{f\supset v}P_t(G_{t(v,f),v}G_{v,s(v,f)},H_{vf}).
\end{align}
Remember also that one of the edge $SL(2,\mbb C)$ integrals is dropped to regularize the amplitude.

The glueing of vertex amplitudes at each face is done integrating the vertices against a anti-holomorphic face amplitude with a suitable measure. For the standard choice \eqref{holonomy full}, the glueing function must be a complexified heat kernel, and the measure is the one resolving the identity \eqref{resolution loop}. Thus the spinfoam amplitude of a 2-complex $\sigma$ in the holomorphic representation reads
\begin{align} \label{holomorphic full}
Z=\int_{\SUtwoC}\prod_{vf} d\Omega(H_{vf}) \prod_v W_t(H_{vf})\prod_f \overline{K_{V(f)t}(\prod_{v\subset f}H_{vf})},
\end{align}
where the heat kernel parameter of the face amplitude\footnote{The definition of the single-variable heat kernel $K_t(h)$ should be clear. In terms of the two-variable heat kernel, it is $K_t(h):=K_t(\mbb 1,h)$.} must have the multiplicity of the number of vertices $V(f)$ in the face $f$. Notice that the face amplitude is anti-holomorphic, which explains the complex conjugation in the last expression. Thus the local amplitudes are holomorphic or anti-holomorphic, whereas the full amplitude is not. Nevertheless we shall call the expression \eqref{holomorphic full} the holomorphic amplitude, for simplicity.

This last expression for the spinfoam amplitude can be shown to be completely equivalent to the original holonomy representation \eqref{holonomy full}. The proof is straightforward, using the aforementioned identity resolution with coherent states.

This new representation is suitable for the analysis of the constraints in phase space, that is the dynamics of the theory. We should expect that the phase space constraints found from the spinfoam formalism are strictly related to the classical constraints of the canonical theory, in particular to the Hamiltonian constraint.

One possibility explored in this paper is the study of the holomorphic partial amplitude, defined as the integrand in \eqref{holomorphic full}, namely the amplitude at \emph{fixed} values of the variables $H_{vf}$, in the small $t$ regime. The parameter $t$ determines the semiclassicality of coherent states.

Thus we are interested in the peakedness properties of the partial amplitude
\begin{align}\label{holomorphic partial}
Z_t(H_{vf})=\prod_v W_t(H_{vf})\prod_f K_{V(f)t}(\prod_{v\subset f}H_{vf}),
\end{align}
in the semiclassical regime $t\rightarrow 0$, and parametrization \eqref{polar}. Physically, this regime corresponds to looking at certain coherent Feynman histories (semiclassical quantum space-times with large individual areas) and determining which ones are suppressed by quantum interference and which are not. By the usual interference mechanism of quantum mechanics, we expect the classical theory to emerge in the semiclassical regime.

A minor technical point that we have to remember in the subsequent semiclassical analysis is that the coherent transform by which we obtained the last formula \eqref{holomorphic partial} was performed with respect to un-normalized states, and the norm depends on $t$. We recall that the norm of the single-loop coherent state $\Psi^t_H$ is exponentially growing as $\sim \exp(|X|^2/4t^{\beta+2})$, and similarly for the general graph. Even though we could have used the normalized states, we prefer the un-normalized ones in order to preserve the holomorphicity and anti-holomorphicity of the local amplitudes, and also to keep formulas as simple as possible. Moreover, we will always discard the un-interesting case where some of the fluxes are vanishing. Taking the behavior of the norm into account, the relevant definition of peakedness for the holomorphic amplitude is the following.
\begin{definition}\label{definition: suppression}
We say that the holomorphic partial amplitude $Z_t(H_{vf})$, where
\begin{align}
H_{vf}=h_{vf}e^{iX_{vf}/t^\beta},\quad X_{vf}\neq 0,\quad \beta>0,
\end{align}
for all the wedges of the 2-complex, is suppressed for $t\rightarrow 0$ if the following behavior holds,
\begin{align}
\big(\prod_{vf}\exp{\frac{-|X_{vf}|^2}{4t^{\beta+2}}}\big) Z_t(H_{vf})=\mathcal O(t^\alpha),\quad \forall \alpha>0.
\end{align}
\end{definition}
In a similar way, we will speak about non-suppressed vertex and face amplitudes separately, meaning that we have multiplied the amplitudes by the appropriate exponential factor.

\section{Semiclassical vertex geometry} \label{section: semiclassical vertex}
A bivector is an element in $\bigwedge^2\Mink$. A simple bivector\footnote{Any non-simple bivector admits a unique orthogonal decomposition into two simple bivectors.} $B$ has the form $B=A\wedge C$, where $A$ and $C$ are 4-vectors. Geometrically, a simple bivector can be thought as an oriented 2-plane segment spanned by $A$ and $C$. It has information about the area of the plane segment, the 2-plane where the segment lies, and the orientation of the 2-plane. The Minkowski metric with signature $-+++$ is used to lower the indices and compute the scalar product of bivectors,
\begin{align}
A\cdot B:=A^{IJ}B_{IJ},
\end{align}
thus the sign of $A\cdot A$ characterizes $A$ as a space-like ($>0$), or time-like ($<0$) bivector. The standard Hodge map $*$ acts on bivectors as 
\begin{align}
(*B)^{IJ}:=\frac{1}{2}\epsilon^{IJ}_{\;\;\;\, KL}B^{KL}. 
\end{align}
We recall also that there is a vector space isomorphism between bivectors and the Lorentz algebra, in such a way that any $SO^{+}(1,3)$ Lorentz transformation can be generated by exponentiation of a bivector with the second index lowered. In other words, an element in the Lorentz algebra can be written as
\begin{align}
L=B\eta,
\end{align}
where $B$ is a bivector, and $\eta$ is the flat metric tensor. An association between the spinfoam flux label $X$ of the coherent states and a bivector can be done as in the following
\begin{definition}[Bivectors in time-gauge] \label{bivectors time-gauge}
A flux variable $X\in su(2)$ defines the space-like simple bivector 
\begin{align}
B(X):=\frac{1}{2}|X|*(1,\hat X)\wedge (1,-\hat X).
\end{align}
The bivector is orthogonal to the reference time-like direction $\mathcal N:=(1,0,0,0)$, and we say it is in the time-gauge.
\end{definition}
The bivector $B(X)$ just defined is a plane segment `at rest' in Minkowski space, with area
\begin{align}
|B(X)|:=\sqrt{B(X)\cdot B(X)}=|X|.
\end{align}
Now we are ready to state the main result on the holomorphic vertex amplitude. We have the following
\begin{proposition}[Asymptotic vertex amplitude] \label{proposition: vertex}
$W_t(H_{vf})$ is non-suppressed for small $t$ if and only if the following relations hold for the holonomy-flux labels $H_{vf}$. There exist $SL(2,\mbb C)$ elements $G_{ve}$ and real parameters $\xi_{vf}$ such that
\begin{align}\label{wedge relation}
U(h_{vf}^{-1}G_{t(ef),v}G_{v,s(vf)})=e^{-\xi_{vf}(\gamma+*)B(X_{vf})},
\end{align}
and the fluxes close to zero at each edge,
\begin{align}\label{closure}
\sum_{f_\text{in}\supset e}X_{vf}-\sum_{f_\text{out}\supset e}h_{vf}\triangleright X_{vf}=0,
\end{align}
where $f_\text{in}$ are the faces that induce on $e$ an ingoing orientation at the vertex, and $f_\text{out}$ the outgoing ones.
\end{proposition}
In the equation \eqref{wedge relation} we used the identification of bivectors with elements in the Lorentz algebra via the flat metric tensor, and the Lie algebra exponential map. The map $U$ is the projection of $SL(2,\mbb C)$ on the $(1/2,1/2)$ finite-dimensional 4-vector representation, namely it is just the covering map of the Lorentz group. In equation \eqref{closure} $h$ acts in the spin-1 3-vector representation. Remember also that in this paper we work only with EPRL-integrable graphs, and this choice is understood in the previous proposition.
We shall see that this proposition implies the existence of a bivector geometry at the vertex $v$, defined as follows.
\begin{definition}[Bivector geometry] An assignment of bivectors $B_{f}(v)$ ($f\supset v$) at a vertex $v$ is called a bivector geometry if the following relations hold.
\begin{itemize}
\item Closure: for every edge $e\supset v$,
$$\sum_{f\supset e}\epsilon_{vef}B_f(v)=0.$$
\item Simplicity: for every face $f\supset v$,
$$B_f(v)\wedge B_f(v)=0.$$
\item Cross-simplicity: for every two faces $f,f'\supset v$ bounded by the same edge $e\subset f,f'$,
$$B_f(v)\wedge B_{f'}(v)=0.$$
\end{itemize}
\end{definition}
We recall that $B_f(v)$ being simple means, equivalently, that it can be written as the exterior product of two 4-vectors. Cross-simplicity states that for any two faces that share the same edge, also the sum $B_f(v)+B_{f'}(v)$ is simple.

Now we want to check that the proposition \ref{proposition: vertex} implies a bivector geometry. Notice that by \eqref{wedge relation}, two bivectors associated to the same face must agree upon transport at the vertex (or the amplitude is suppressed). In other words, we are able to define the following bivectors `in the frame of the vertex',
\begin{align} \label{transport constraint}
B_{f}(v):=G_{v,s(vf)}\triangleright B(X_{vf})=G_{v,t(vf)}h_{vf}\triangleright B(X_{vf}),
\end{align}
obtained by transport of the boundary data in a common frame at the vertex $v$. Here $SL(2,\mbb C)$ acts in the 4-vector representation on the bivectors, or in the adjoint representation if we think bivectors as Lorentz algebra elements. The right equality in \eqref{transport constraint} holds for non-suppressed vertex amplitudes by proposition \ref{proposition: vertex}.

By the construction \eqref{bivectors time-gauge}, the vertex bivectors \eqref{transport constraint} are simple, and by \eqref{closure} close to zero at each edge because the fluxes $X_{vf}$ do. Moreover, the bivectors are cross-simple at each edge $e$, since those ones in the time-gauge are cross-simple by construction. In fact, they satisfy a constraint which is stronger than cross-simplicity: the vertex bivectors at the edge $e$ all lie in the space-like 3-plane orthogonal to $G_{ve}\triangleright\mathcal N$.

Thus we have shown the following corollary of proposition \ref{proposition: vertex}.
\begin{corollary}\label{corollary: bivector}
The holomorphic vertex amplitude is non-suppressed only if there are $SL(2,\mbb C)$ elements $G_{ve}$ such that $H_{vf}$ determines a bivector geometry $B_f(v)$ via \eqref{transport constraint}.
\end{corollary}
Notice that this corollary is almost equivalent to the proposition \ref{proposition: vertex}, but it is weaker. Proposition \ref{proposition: vertex} contains an information that is missing in the bivector geometry equations. The extra information is precisely the proportionality of the two coefficients multiplying $B(X_{vf})$ and $* B(X_{vf})$ in the r.h.s. of \eqref{wedge relation}. This extra requirement was already discussed in some detail in reference \cite{Bianchi:2010mw}, in a different language. We shall discuss in detail its physical meaning in section \ref{section: Ashtekar twists}.

\section{The face amplitude: connecting vertices}\label{section: semiclassical face}
In this section we study the constraints imposed by the face amplitude defined in the formula \eqref{holomorphic full} of the holomorphic partition function.

The analysis is independent: we do not impose at this stage the constraints found in the analysis of the vertex amplitude. As expected, it turns out that the anti-holomorphic face amplitude is responsible of the gluing of vertices. In particular, it implies the area-matching constraint.

The result is the following
\begin{proposition}\label{proposition: face amplitude}
The anti-holomorphic face amplitude 
\begin{align}
\overline{K_{V(f)t}(\prod_{v\subset f}H_{vf})}
\end{align}
is non-suppressed for $t\rightarrow 0$ if and only if the following relations hold. For each edge $e\subset f$ the glueing equation is satisfied,
\begin{align} \label{glueing constraint}
X_{t(e)f}=h_{s(e)f}\triangleright X_{s(e)f},
\end{align}
and the loop Ashtekar holonomy is trivial,
\begin{align} \label{flatness constraint}
\prod_{v\subset f}h_{vf}=1.
\end{align}
\end{proposition}

As expected, the first equation \eqref{glueing constraint} is the constraint that imposes the \emph{glueing} of the vertices, meaning that the fluxes match at the interface of two vertices. To see this, let us explain the formula considering two adjacent vertices $v$, $v'$ that bound an edge $e$ in the 2-complex. Consider one of the oriented faces $f\supset e$, and say the induced orientation on $e$ is such that $v$ is the source vertex, or $\epsilon_{vef}=-\epsilon_{v'ef}=1$. Notice that on the edge $e$ we have always two fluxes defined for the same face. A source flux in the boundary of one vertex, in this case $X_{v'f}$, and a target flux $h_{vf}\triangleright X_{vf}$ in the other vertex. The condition \eqref{glueing constraint} requires the two to be equal. In particular we see that their equality implies that $|X_{vf}|=|X_{v'f}|$.

Since the modulus of the flux is proportional to the area of the 2-cell $f$, we have just seen that the glueing constraint implies in particular the \emph{area-matching} constraint at the edges\footnote{The area matching constraint is stronger then just a semiclassical equation. It holds exactly at the quantum level on the quantum numbers of the area, namely on the spin variables: $j_{vf}=j_f$ for all $v\subset f$.}. Thus we have that for each face, all the quantities $|X_{vf}|$ coincide, or the face amplitude is suppressed.

We can also give an interpretation of the glueing constraint in terms of the geometry of polyhedra, provided that we use the closure constraint \eqref{closure} of the holomorphic vertex amplitude. Indeed when the glueing constraint holds, we can define the following edge-face area 3-vectors,
\begin{align}
A_{ef}:=
\begin{cases}
X_{v'f}=h_{vf}\triangleright X_{vf},\quad \epsilon_{evf}=+1\\
X_{vf}=h_{v'f}\triangleright  X_{v'f},\quad \epsilon_{evf}=-1
\end{cases}
\end{align}
for arbitrary orientations, where $A_{ef}$ is a 3-vector using the isomorphism $su(2)\simeq \mbb R^3$. Provided the following closure constraint holds,
\begin{align}
\sum_{f\supset e} \epsilon_{vef}A_{ef}=0,
\end{align}
a theorem by Minkowski \cite{Bianchi:2010gc} implies that there exists a unique polyhedron in $\mbb R^3$, up to translations, such that the vectors $\epsilon_{vef}A_{ve}$ are the external normals to the faces of the polyhedron, normalized to the area of the faces. Notice that if we require $\epsilon_{vef}A_{ve}$ to be the internal normals, we determine the parity-related polyhedron. Thus we find that thanks to the glueing constraint the vertex labels $H_{vf}$ and $H_{v'f}$ determine the same polyhedral geometry at the edge $e$.

So far so good. What about the last condition \eqref{flatness constraint}? Let us say that it implies an unexpected constraint on space-time curvature for all values of the Barbero-Immirzi parameter, except for $\gamma=0$ where \eqref{flatness constraint} becomes redundant. In fact, it is analogous to the `flatness constraint' agued in the spin-intertwiner spinfoam representation in \cite{Bonzom:2009hw}, here analyzed in detail in terms of classical geometric variables thanks to the new holomorphic representation.

To see this important point more in detail we postpone the full analysis of the anti-holomorphic face amplitude to the special case of a simplicial 2-complex, which is technically easier.
\section{A special case: the 4-simplex} \label{section: 4-simplex}

From this point of the discussion we specialize our analysis of the Lorentzian EPRL model to a simplicial 2-complex. We start with the analysis of the simplicial holomorphic vertex amplitude. The graph $\Gamma_v$ of a simplicial vertex $v$ is the complete graph with five nodes. Thus every vertex bounds five edges and ten faces.

For our analysis, we first need some definitions in classical simplicial geometry.

\begin{definition}[Geometric 4-simplex] A geometric 4-simplex $\sigma_v$ is the convex hull of five points in $\Mink$ not all of which lie in the same 3-plane. 
\end{definition}
In this paper we consider only 4-simplices where all the triangles are space-like. This choice is always understood. The reason is that the representations of the Lorentz group used to define the spinfoam model restrict automatically the partition function to such space-like geometries.

We label the standard orientation $dx^0\wedge dx^1\wedge dx^2\wedge dx^3$ of $\Mink$ with a sign $\mu_v=1$, and the opposite one with $\mu_v=-1$. The Hodge duality requires an orientation. Thus we can define the oriented Hodge duality map, in terms of the standard orientation, as $*_{\mu_v}:=\mu_v *$.
\begin{definition}[Oriented geometric 4-simplex] An oriented geometric 4-simplex $(\sigma_v,\mu_v)$ is a geometric 4-simplex together with an orientation $\mu_v$ of $\Mink$. 
\end{definition}
The orientation $\mu_v$ provided to $\Mink$ induces an orientation on the 3-dimensional boundary of the 4-simplex, formed by five tetrahedra, which in turn induces an orientation in the boundary of each tetrahedron. The bivectors of a 4-simplex can be identified with its oriented triangles, defined as follows.
\begin{definition}[Area bivectors] \label{area bivectors}
The area bivectors $B_{f}(\sigma_v)$ of an oriented geometric 4-simplex $(\sigma_v,\mu_v)$ are defined as
\begin{align}
B_{f}(\sigma_v,\mu_v):=\frac12 a_{f}*_{\mu_v}\frac{N_{s(vf)}(v)\wedge N_{t(vf)}(v)}{|N_{s(vf)}(v)\wedge N_{t(vf)}(v)|},
\end{align}
where $a_f=|B_f|$ is the area of the triangle $f$ computed with the Minkowski metric, and $N_e(v)$ is the external unit normal to the tetrahedron $e$.
\end{definition}
Notice that in the previous formula the orientation of the bivector, or the triangle $f$, is the one induced in the boundary of the source tetrahedron $s(vf)$, which in turn has an orientation pull-backed from $\mu_v$, as discussed previously. 

The following definition of non-degenerate bivector geometry is useful to state the precise relation between bivectors and simplices, which is the content of the Barrett-Crane theorem.
\begin{definition}[Non-degenerate bivector geometry] A bivector geometry $B_{f}(v)$ is said non-degenerate if on each edge the set of bivectors
$$B_{f}(v),\quad f\supset e$$
spans a 3-plane, and for any two faces $f$, $f'$ that do not share an edge, the bivectors
$$B_{f}(v),B_{f'}(v)$$
span $\Mink$.
\end{definition}

The following theorem is a straightforward generalization to Lorentzian signature of the original theorem \cite{Barrett:1997gw} for the Euclidean case. Furthermore, we have adapted the theorem to our definitions concerning simplicial geometry with orientations.
\begin{theorem}[Barrett-Crane theorem]Given an oriented Lorentzian geometric 4-simplex $(\sigma_v,\mu_v)$, its area bivectors $B_f(\sigma_v,\mu_v)$ satisfy the bivector geometry constraints and are non-degenerate. Conversely, given a non-degenerate bivector geometry $B_{f}(v)$ there exists a unique (up to inversion $x^\mu\rightarrow -x^\mu$ and translations) oriented geometric 4-simplex $(\sigma_v,\mu_v)$ such that
\begin{align}
B_{f}(\sigma_v,\mu_v)=B_{f}(v).
\end{align}
\end{theorem}
The first part of the theorem is pretty obvious geometrically. The most important part is the second one, namely the reconstruction part: any non-degenerate bivector geometry arises from a unique 4-simplex, up to obvious symmetries of the bivectors.

At this point we observe that the proposition \ref{proposition: vertex} can be refined a little bit, in order to extract information about the uniqueness of the elements $G_{vf}$, and determine the various sectors of the possible solutions. 

This useful information is in the following
\begin{lemma} \label{lemma: classification vertex solutions}
In the proposition \ref{proposition: vertex}, if there exist elements $G_{ve}$ which determine a non-degenerate bivector geometry, they are unique up to a rigid $G_v\in SL(2,\mbb C)$ transformation at all $e\supset v$, 
\begin{align}
G_{ve}\rightarrow G_v G_{ve},
\end{align}
a parity transformation at all $e\supset v$,
\begin{align}
G_{ve}\rightarrow G^*_{ve}:=(G_{ve}^\dagger)^{-1},
\end{align}
and a spin lift symmetry\footnote{The spin-lift symmetry in $SL(2,\mbb C)$ should not be confused with the space-time inversion symmetry. Space-time inversion $x^\mu\rightarrow -x^\mu$ is a $SO(1,3)$ Lorentz transformation not connected to the identity, thus not in $SL(2,\mbb C)$ under the covering map.} at some $e\supset v$,
\begin{align}
G_{ve}\rightarrow -G_{ve}.
\end{align}
\end{lemma}

By this lemma, a 4-simplex is always determined up to $SO^+(1,3)$ rigid transformations, parity, inversion and translations. Thus it is determined up to general Poincar\'e transformations, as expected. This shows that the full group of local space-time symmetries arises in the classical simplicial geometry.

Other useful facts follow from lemma \ref{lemma: classification vertex solutions}. For instance, it implies that the non-degeneracy of the bivector geometry at a vertex $v$ is an intrinsic property of the flux-holonomy boundary variables $H_{vf}$. In other words, either they imply a degenerate geometry, or they imply a non-degenerate geometry. Thus the two sectors do not mix for fixed holomorphic vertex labels $H_{vf}$.

Moreover, we can see that there are in principle two parity sectors associated to a fixed set of vertex boundary labels. The two sectors are characterized by the sign $\mu_v$ of the reconstructed orientation, as expected since parity is not orientation-preserving. However, the Ashtekar holonomy selects only one parity sector as discussed in section \ref{section: parity}.

Thus we may summarize what we have learnt by saying that the labels of the holomorphic vertex determine three mutually exclusive cases, or sectors. These are
\begin{itemize}
\item a) \emph{non-degenerate geometric}: a non-degenerate bivector geometry
\item b) \emph{degenerate geometric:} a degenerate bivector geometry 
\item c) \emph{non-geometric:} do not determine a bivector geometry
\end{itemize}
In the case c) the vertex amplitude is suppressed at small $t$ by corollary \ref{corollary: bivector}, this is the non-geometric sector. The case b) do not necessarily imply a suppression, this is the degenerate geometric sector. Though interesting, we shall not consider this sector in the paper. The case a) is the non-degenerate geometric sector, our main interest.

\section{Ashtekar connection in spinfoams}\label{section: Ashtekar twists}
So far we did not consider the full constraints implied by proposition \ref{proposition: vertex}. This means that there can be coherent state labels determining a non-degenerate vertex geometry that nevertheless suppress the vertex amplitude for $t\rightarrow 0$. This may sound weird, but we have to remember that the variables of the holomorphic representation contain more information than the more common spin-intertwiner labels, where similar geometric constraints have been found. This extra input can be identified with the variable conjugate to the areas \cite{Bianchi:2009ky,Freidel:2010aq}, which is an extrinsic curvature scalar, as we shall see in a moment.

The boundary of a 4-simplex has an intrinsic as well as an extrinsic geometry. The extrinsic geometry of a 3-surface describes the embedding properties of the surface in four dimensions. For the 3-boundary of a 4-simplex, it is provided by the 4-dimensional dihedral angles. These determine the amount of bending of the 3-boundary at each triangle. As in the Lorentzian Regge calculus, the dihedral angle $\Theta_{vf}$ at the triangle $f^*$ of a 4-simplex with space-like boundary is defined via the scalar product of the external normals $N_{s(vf)}(v)$, $N_{t(vf)}(v)$ of the two tetrahedra $s(vf)$, $t(vf)$ that share the triangle $f^*$, up to a sign. We have
\begin{align}
N_{s(vf)}(v)\cdot N_{t(vf)}(v)=:\cosh\Theta_{vf}.
\end{align}
The sign of $\Theta_{vf}$  is taken positive when both normals are future-pointing or past-pointing, negative when one is future-pointing and the other past-pointing. Let us come back to the main question of this section.

As we already observed at the end of section \ref{section: semiclassical vertex}, corollary \ref{corollary: bivector} is weaker then proposition \ref{proposition: vertex}. We have to impose one extra condition in order to recover the full constraints enforced by the holomorphic vertex amplitude. What is exactly this requirement, and its physical interpretation? The extra requirement is in fact that the torsion of the Ashtekar holonomies must match the extrinsic curvature of the boundary of the 4-simplex. To see this, we first need the decomposition of the $SL(2,\mbb C)$ transports into a rotation and a dihedral boost, as stated in the following
\begin{lemma} \label{lemma: natural decomposition of GG}
Given a set of vertex labels $H_{vf}$ such that there exist elements $G_{vf}$ which determine an oriented 4-simplex $(\sigma_{v},\mu_v)$, for each wedge we must have
\begin{align}\label{natural 1}
U(g^{-1}_{vf}G_{t(vf),v}G_{v,s(vf)})=e^{-\Theta_{vf} *_{\mu_v}B(X_{vf})},
\end{align}
or, equivalently,
\begin{align}\label{natural 2}
U(G_{v,t(vf)}g_{vf}G_{s(vf),v})=e^{\Theta_{vf} N_{s(vf)}\wedge N_{t(vf)}},
\end{align}
where $g_{vf}$ is the 3-dimensional spin holonomy, $N_{e}(v)$ is the unit external normal of the tetrahedron $e^*$ in the reconstructed 4-simplex $\sigma_v$, and $\Theta_{vf}$ the Lorentzian dihedral angle at the triangle $f^*$.
\end{lemma}
We recall the definition of the 3-dimensional Levi-Civita holonomy. From the eight labels $H_{vf}$, with $f\supset e,e'$, associated to two nodes $e$, $e'$, we can reconstruct two tetrahedra in $\mbb R^3$ using the Minkowski theorem. The Levi-civita holonomy from the source node to the target node is the unique $O(3)$ rotation that maps the triangle $f^*$ in the source tetrahedron to the triangle $f^*$ in the target tetrahedron, and the unit external normals to the triangles in the antiparallel configuration. In fact, since the normals are taken all external or all internal, and thus the two tetrahedra are consistently oriented, the transformation belongs to $SO(3)$. The $SU(2)$ transformation $g_{vf}$ is the 3-dimensional spin holonomy, namely the Levi-Civita holonomy up to the spin lift ambiguity.

In a similar way, the Ashtekar holonomy can be always split in two parts, the 3-dimensional spin holonomy and a `twist', in the following way,
\begin{align}\label{Ashtekar split}
h_{vf}=g_{vf}e^{\alpha_{vf}\hat X_{vf}\cdot\vec\tau}.
\end{align}
Comparing lemma \ref{lemma: natural decomposition of GG} with the equation $\eqref{Ashtekar split}$, we were able to write the 4-dimensional spin holonomy and the 3-dimensional Ashtekar holonomy in a similar fashion. They have a common part, which is the 3-dimensional spin holonomy $g_{vf}$. Substituting the expressions for $h_{vf}$ and the result of lemma \ref{lemma: natural decomposition of GG} into \eqref{wedge relation}, we must have
\begin{align}
e^{-\alpha_{vf} B(X_{vf})-\mu_v\Theta *B(X_{vf})}=e^{-\gamma\xi_{vf} B(X_{vf})-\xi_{vf}*B(X_{vf})},
\end{align}
namely
\begin{align}
\begin{cases}
\gamma\xi_{vf} = \alpha_{vf}+4k\pi, \\
\xi_{vf} = \mu_v \Theta_{vf},
\end{cases}
\end{align}
where $k$ is an integer. Notice that from \eqref{Ashtekar split} the correct periodicity of $\alpha_{vf}$  is $4\pi$. Thus we have found the following
\begin{proposition} \label{proposition: extra constraint}
A holomorphic vertex amplitude in the non-degenerate sector is non-suppressed for $t\rightarrow 0$ if and only if
\begin{align} \label{extra constraint}
\alpha_{vf} = \pm\gamma\Theta_{vf}\mod 4\pi,
\end{align}
for each face $f\supset v$, where $\alpha_{vf}$ and $\Theta_{vf}$ are the torsion of the Ashtekar holonomy and the dihedral angle at the triangle $f^*$ of the reconstructed 4-simplex $\sigma_v$, respectively. The $\pm$ sign is the 4-simplex orientation $\mu_v$.
\end{proposition}

The proposition \ref{proposition: extra constraint} uncovers the geometric meaning of the extra condition: the difference between the $SL(2,\mbb C)$ holonomy and the Ashtekar holonomy is in the way they code the extrinsic geometry of space. The first \emph{bends} the tetrahedra in four dimensions, creating a non-trivial extrinsic geometry, whereas the second cannot bend, being a pure 3-dimensional rotation. However, the Ashtekar holonomy is smart: it performs a twist, or \emph{torsion}, of one tetrahedron with respect to the other, via a $U(1)\subset SU(2)$ rotation about the normal of the common triangle. The torsion angle $\alpha$ codes the 4-dimensional dihedral angle, and must match the `true' dihedral angle computed out of the 4-simplex geometry, otherwise the holomorphic vertex amplitude is suppressed. The matching of torsion and extrinsic curvature is the meaning of the constraint \eqref{extra constraint}.

In the language of the $SU(2)$ spin quantum numbers, this constraint is exactly equivalent to the requirement that the rapidly oscillating phase $\sim e^{ij\alpha}$ in the boundary state must cancel a similar phase factor in the dynamics in order to have a good semiclassical behavior, a well-known property of quantum-mechanical wave-packets first advocated in the spinfoam setting by the Rovelli's ansatz \cite{Rovelli:2005yj}, and crucial in the graviton propagator calculations \cite{Bianchi:2006uf,Alesci:2007tx,Bianchi:2009ri}.

The relation \eqref{extra constraint} together with \eqref{Ashtekar split} are the spinfoam analogous of the formula
\begin{align} \label{LQG Ashtekar}
A=\Gamma+\gamma K,
\end{align}
for the Ashtekar connection with real Barbero-Immirzi parameter in classical general relativity. Here $\Gamma$ is the 3-dimensional spin connection and $K$ the extrinsic curvature of a Cauchy 3-surface. In general relativity, the Ashtekar connection have torsion, and torsion \emph{is} the extrinsic curvature of space. Notice the correct presence of $\gamma$ in front of the extrinsic curvature in the spinfoam expression \eqref{extra constraint}. This concludes the semiclassical analysis of the holomorphic vertex amplitude.

The condition \eqref{extra constraint} was first found in the holomorphic representation in \cite{Bianchi:2010mw}. We have deepened the analysis of the physical meaning of this condition in terms of the torsion part of the Ashtekar connection.
\section{The role of parity and time-reversal}\label{section: parity}
The parity transformation (see lemma \ref{lemma: classification vertex solutions}) relates the two parity sectors of the bivector geometries. To see this in more detail, consider the spatial inversion $P:\vec x\rightarrow \vec x$, which acts on the 4-vectors in $\Mink$, and by extension on the bivectors. Notice that a bivector $B$ that lies in the 3-plane orthogonal to $\mathcal N$ is parity-invariant. Using this, it is easy to check that the parity transform of a boosted bivector can be written simply in terms of the action of another `starred' $SL(2,\mbb C)$ transformation as
\begin{align} \label{parity of bivector}
PG\triangleright B=G^*\triangleright B,
\end{align}
where we recall that the notation $G^*$ stands for the conjugate-inverse of a $SL(2,\mbb C)$ matrix. This descends immediately from a well-known result in representation theory, often used in particle theoretical physics: the defining representation of $SL(2,\mbb C)$ is the left-handed representation. The right-handed representation is obtained by taking the conjugate-inverse of the elements.

We notice also that by the inversion symmetry of a bivector, the previous equation holds for the time reversal $T:x^0\rightarrow -x^0$. Since the space-time inversion $PT:x^\mu\rightarrow -x^\mu$ leaves any bivector invariant, we have
\begin{align}\label{time-reversal of bivector}
TG\triangleright B=PTTG\triangleright B=PG\triangleright B=G^*\triangleright B.
\end{align}

We shall see how this relates to the holomorphic vertex amplitude. Consider flux-holonomy boundary variables $H_{vf}$ such that there are elements $G_{vf}$ that determine a non-degenerate bivector geometry $B_f(v)$, and in turn an oriented 4-simplex $(\sigma_v,\mu_v)$, up to inversion and translations. Then, by equation \eqref{parity of bivector}, or equivalently \eqref{time-reversal of bivector}, the elements $G^*_{vf}$ determine a non-degenerate bivector geometry as well, the relation between the two being
\begin{align}
B'_f(v)=PB_f(v)=TB_f(v).
\end{align}
In turn, the new bivector geometry determines an oriented 4-simplex which is related to the previous one by
\begin{align}\label{parity-related 4-simplex}
(\sigma'_v,\mu'_v)=(P\sigma_v,-\mu_v),
\end{align}
or equivalently, using the inversion ambiguity, by
\begin{align} \label{time-reversed 4=simplex}
(\sigma'_v,\mu'_v)=(T\sigma_v,-\mu_v),
\end{align}
where of course $\sigma'_v$ in \eqref{time-reversed 4=simplex} is not the same 4-simplex $\sigma'_v$ of \eqref{parity-related 4-simplex}. Thus $P$ and $T$, which are the discrete Lorentz transformations that do non preserve the space-time orientation, act on the 4-simplex \emph{and} on its orientation.

Given this, we realize an important feature of the holomorphic vertex amplitude. The vertex amplitude brakes the parity symmetry selecting a single orientation $\mu_v$. Indeed we recall that the $\pm$ sign in the equation \ref{proposition: extra constraint} relating the torsion of the Ashtekar holonomy to the dihedral angle is precisely the orientation $\mu_v$. As we have seen previously, the flux-holonomy variables $H_{vf}$ determine a 4-simplex up to general Poincar\'e transformations. However, `half' of these Poincar\'e-related 4-simplices have orientation $\mu_v=+1$, the other `half' $-1$.

Thus we can prescribe the Ashtekar torsion in order to match either the dihedral angle, or minus the dihedral angle. In this way only one parity sector of the family of reconstructed 4-simplices is allowed for a \emph{fixed} set of boundary variables. The symmetry between the two parity sectors is restored when we consider the set of \emph{all} flux-holonomy boundary variables. The theory does not distinguish them thanks to the integral over $H_{vf}$ that we have to perform so as to recover the full amplitude \eqref{holomorphic full}.


\section{Glueing of 4-simplices} \label{section: 4-simplex glueing}
We specialize the analysis of the face amplitude to the simplicial setting. In this section we consider only the first \eqref{glueing constraint} of the anti-holomorphic face amplitude equations. Together with the results for the holomorphic simplicial vertex amplitude, this enables us to state a reconstruction theorem for the full triangulation. 

\begin{definition}[Regge triangulation] A Regge triangulation is a 4-dimensional simplicial complex, together with a piecewise flat Lorentzian metric such that the 4-cells are isometric to geometric 4-simplices in $\Mink$.
\end{definition}
In a Regge triangulation, the metric is locally flat everywhere except on the 2-cells, which are the triangles of the triangulation. A Regge triangulation can be constructed via a pairwise glueing process of 4-simplices that have compatible geometries. 

Notice that we can always choose Cartesian charts $\phi_v$, each one covering a single 4-cell $v$. The Cartesian charts determine a tetrad, or inertial reference frames, namely four orthonormal tangent vector fields $e_I,\,I=0,\ldots,3$, one time-like and three space-like defined locally via
\be\label{tetrad}
e_I:=\delta_I^\mu\frac{\partial}{\partial x^\mu}.
\ee
With this choice of charts, the transition functions from the 4-cell $v$ to an adjacent 4-cell $v'$ are Poincar\'e transformations, that we may write as
\begin{align}
\mathcal P_{v'v}:=U_{v'v}\times \mathcal T_{v'v},
\end{align}
where $U_{v'v}$ is the rotation part, namely a Lorentz transformation, and $\mathcal T_{v'v}$ the translation part. The rotation part is called Levi-Civita holonomy. It is the parallel-transport matrix for vectors from the tangent space at $v$ to the one at $v'$.

In general, the Levi-Civita holonomy is a $O(1,3)$ transformation. If there is a choice of charts such that all Levi-Civita holonomies are $SO(1,3)$ transformations, the Regge triangulation is defined to be orientable.

Viewed in the respective inertial reference frames, the 4-cells are represented as geometric 4-simplices $\sigma_v$ in $\Mink$. We recall our assumption that the boundaries of these 4-simplices are space-like. Also, the 3-cell shared by two 4-cells has two different representations as a geometric tetrahedron in $\Mink$, the one in the chart $\phi_v$ and the one in the adjacent chart $\phi_{v'}$. The Levi-Civita holonomy from $v$ to $v'$ is thus the unique $O(1,3)$ Lorentz transformation $U_{v'v}$ that, together with a translation, maps the tetrahedron in one frame into the corresponding tetrahedron in the other frame, namely
\begin{equation}\label{tetrahedraglueing}
U_{v'v}\phi_v(x)=\phi_{v'}(x),
\end{equation}
for every point $x$ of the 3-cell, and brings the external unit normals in the anti-parallel configurarion,
\be\label{normalsglueing}
U_{v'v}N_e(v)=-N_e(v').
\ee

Only in a flat space-time region centered at $f$ we can choose a single Cartesian chart to cover all the 4-cells in the `loop' $v\supset f$, which yields trivial transition functions. In order to extract the information about space-time curvature from the transition functions, we can compute the loop Levi-Civita holonomy around a face $f$, which is the composition of the parallel transports in the closed sequence of edges $e$ that form the boundary of the face $f$. We have to specify also the orientation of the loop and the base vertex. We define the loop holonomy based at $v$ as
\begin{equation}
U_f(v):=U_{v,s(e_m)}\ldots U_{t(e_1),v},
\end{equation}
where $(e_1,\ldots,e_m)$ is the cyclic sequence of edges in a face bounded by $m$ edges, and $v=s(e_1)=t(e_n)$. A loop Levi-Civita holonomy around $f$ different from the identity signals the presence of a scalar curvature on the 2-cell $f^*$. The relevant gauge-invariant curvature information is the deficit angle.
\begin{definition}[Deficit angle] The deficit angle at a 2-cell $f^*$ is the rapidity $\theta_f$ of the boost part of the loop Levi-Civita holonomy $U_f(v)$, where the orientation of the boost is from $N_{s(vf)}(v)$ to $N_{t(vf)}(v)$.
\end{definition}
One can easily check that the deficit angle is independent of: the starting point $v$, the face orientation, the choice of local frames. Thus it is a function only of the metric of the Regge triangulation. The deficit angle at a 2-cell $f^*$ is zero if and only if the metric is flat in the space-time region formed by the 4-cells in the `loop' $v\subset f$.

Another equivalent definition that can be found in the literature is in terms of the 4-dimensional dihedral angles, that we discussed in section \ref{section: 4-simplex}. The deficit angle is simply their sum at a face, namely
\begin{definition}[Deficit angle, bis] The deficit angle at a 2-cell $f^*$ of a Regge triangulation is
\begin{align}
\Theta_f:=\sum_{v\subset f}\Theta_{vf},
\end{align}
where $\Theta_{vf}$ is the dihedral angle at the triangle $f^*$ of the 4-simplex representation of the 4-cell $v$ in a Cartesian chart.
\end{definition}
Observe that the dihedral angles depend only on the Poincar\'e-invariant geometry of the 4-simplex, confirming that the deficit angle is a function only of the metric.

Now we see how the spinfoam semiclassical constraints relate the holonomy-flux variables to a Regge triangulation. We begin with the following
\begin{proposition}[4-simplex glueing]\label{proposition: 4-simplex glueing}
Consider two sets $X_{vf}$ and $X_{v'f}$ of holonomy-flux variables for adjacent vertices $v$ and $v'$. Suppose there exist $SL(2,\mbb C)$ elements $G_{ve}$ and $G_{v'e}$ which determine non-degenerate bivector geometries at $v$ and $v'$ respectively, i.e. oriented 4-simplices $(\sigma_v,\mu_v)$ and $(\sigma_v',\mu_v')$ up to inversion and translations. Then the glueing equation \eqref{glueing constraint} implies that the Levi-Civita holonomy $U_{v'v}$ exists and is given by,
\begin{align}\label{eq: glueing theorem}
U_{v'v}:=\begin{cases}\mu_e U(G_{v'e}G_{ev}),&\mu_v=\mu_{v'}\\
\mu_e U(G_{v'e})PU(G_{ev}),&\mu_v\neq\mu_{v'}
\end{cases}
\end{align}
where $P$ is the spatial inversion with respect to the fiducial $\mathcal N$, $\mu_e:=p_v p_{v'}$ is an overall sign, and $p_v=\pm 1$ parametrizes the inversion ambiguity in the Barrett-Crane 4-simplex reconstruction.
\end{proposition}

A negative overall sign $\mu_e=-1$ in \eqref{eq: glueing theorem} is a space-time inversion. The first important thing to notice is that by the previous proposition the Levi-Civita holonomy has positive determinant if and only if the reconstructed orientations $\mu_v$ and $\mu_{v'}$ agree. As a consequence, the spinfoam $SL(2,\mbb C)$ elements cannot be always interpreted as the 4-dimensional spin-connection, for two reasons: one is the parity insertion \eqref{eq: glueing theorem} in the case of non-matching orientations $\mu_v\neq\mu_{v'}$, the second is the inversion ambiguity $\mu_e$ of the 4-simplex reconstruction.

We have shown that the holomorphic partial amplitude for a simplicial 2-complex is peaked on holonomy-flux Ashtekar variables which realize a collection of geometric 4-simplices in $\Mink$, up to the local symmetries of lemma \ref{lemma: classification vertex solutions}. Moreover, there always exists a Levi-Civita holonomy connecting adjacent 4-simplices, which means that the transition functions of a Regge triangulation are well-defined. Thus we have
\begin{corollary}[Reconstruction of the Regge triangulation] The simplicial holomorphic partial amplitude $Z_t(H_{vf})$ is non-suppressed for $t\rightarrow 0$ only if it determines a Regge triangulation, or a degenerate geometry.
\end{corollary}
In the next section we study in the simplicial setting the last equation implied by the semiclassical limit.
\section{The curvature constraint} \label{section: curvature constraint}
Let us summarize what we have done so far. We have seen in the previous sections that there are two kind of phase space constraints imposed by the holomorphic vertex amplitude at small $t$, namely the transport equation \eqref{wedge relation}, and the closure equation \eqref{closure}, plus the simplicity constraints which hold by construction. The anti-holomorphic face amplitude imposes the glueing constraint \eqref{glueing constraint} at the interfaces of vertices, namely across the edges.

In the case of a simplicial 2-complex, all these constraints determine a Regge triangulation with curvature, in general, as we have seen in the last section \ref{section: 4-simplex glueing}.

Here we analyze the effect of the triviality of the Ashtekar loop holonomy on a simplicial 2-complex, which is our last constraint. This is the second constraint \eqref{flatness constraint} from the anti-holomorphic face amplitude in the $t\rightarrow 0$ limit. Using the splitting \eqref{Ashtekar split}, we can write the constraint on the Ashtekar loop holonomy as
\begin{align}\label{Ashtekar loop collapse}
\prod_{v\subset f}h_{vf} = \prod_{v\subset f} g_{vf}e^{\alpha_{vf}\hat X_{vf}\cdot\vec\tau}=e^{(\sum_{v\subset f}\alpha_{vf})\hat X_{vf}\cdot\vec\tau}=1,
\end{align}
where the first holonomy in the ordered product is labeled with a generic couple $vf$. 

To obtain the equation \eqref{Ashtekar loop collapse}, in particular the second equality, we have used the defining property of the 3-dimensional Levi-Civita $SO(3)$ holonomy. Consider the two tetrahedra $s(vf)^*$ and $t(vf)^*$ in $\mbb R^3$, namely the tetrahedra in the time-gauge constructed from the holonomy-flux coherent state labels. Since $g_{vf}$ maps the triangle $f^*$ in the source tetrahedron $s(vf)^*$ into the triangle $f$ of the target tetrahedron $t(vf)^*$, this implies that for the closed loop around the face $f$ the loop Levi-Civita $SO(3)$ holonomy based at the edge $s(vf)$ stabilizes the triangle $vf$. Thus the loop $SU(2)$ spin holonomy must be the identity, up to the spin-lift ambiguity, namely
\begin{align}
\prod_{v\subset f}g_{vf}=\pm 1,
\end{align}
and by a simple calculation the second equality in \eqref{Ashtekar loop collapse} readily follows. Notice that the sign ambiguity in the previous equation can be always chosen as positive by appropriately defining the splitting \eqref{Ashtekar split}.

Now we use the constraint on the torsion $\alpha_{vf}$ enforced by the vertex amplitude, proposition \ref{proposition: extra constraint}, and we find immediately
\begin{align}
\sum_{v\subset f}\alpha_{vf}=\gamma\sum_{v\subset f}\mu_v\Theta_{vf}+4k\pi=4k'\pi,
\end{align}
where $k$, $k'$ are integers. When the 4-simplex orientations are such that $\mu_v=1$ or $\mu_v=-1$ for all vertices in the loop, we get
\begin{align}
\gamma\Theta=0\mod 4\pi,
\end{align}
which means that for $\gamma\neq 0$ at most countably many deficit angles are allowed. More precisely, the deficit angle must be zero up to multiples of $4\pi/\gamma$ multiples. We are tempted to interpret the non-zero deficit angles as accidental curvatures \cite{Hellmann:2012kz} since their origin is the periodicity in $SU(2)$.

On the contrary, for $\gamma=0$, which corresponds to the Lorentzian version of the Euclidean flipped spinfoam model \cite{Engle:2007qf} without Barbero-Immirzi parameter, the curvature constraint disappears, and the full continuous set of Regge curvatures is allowed! The flipped model is unphysical for it corresponds to the quantization of the sole Holst term of the classical action. However, it is possible that in a suitable small $\gamma$ limit the spinfoam amplitude can allow for a continuous set of non-trivial curvatures. This last observation seems to be in remarkable agreement with the flipped expansion studied in \cite{Magliaro:2011zz,Magliaro:2011dz}.

In the case of general orientations, the constrained quantity at each face is
\begin{align}\label{generalized deficit}
\gamma\sum_{v\subset f}\mu_v\Theta_{vf}=0\mod 4\pi,
\end{align}
namely the generalized deficit angle in the sense of Barrett and Foxon \cite{Barrett:1993db}. Equation \eqref{generalized deficit} still implies that space-time curvature is zero up to the accidental $4\pi/\gamma$ multiples, for the generalized deficit angle has precisely the same geometric content of the Regge deficit angle. The origin of the arbitrary signs $\mu_v$ in \eqref{generalized deficit} is in the fact that the EPRL model aims to quantize an action written in the tetrad formalism, and no restrictions on the tetrad orientations are imposed in the most popular version of the model. Recently, modified models with different behavior with respect to parity were introduced \cite{Engle:2011un,Rovelli:2012yy}.

\section{Conclusions}
We have studied a candidate holomorphic path integral representation for loop quantum gravity, obtained via the Segal-Bargmann transform of the $SU(2)$ holonomy formulation of the Lorentzian EPRL spinfoam model. The transform is defined with respect to the Ashtekar-Lewandowski-Marolf-Mour\~ao-Thiemann adaptation to loop quantum gravity of the Hall coherent states for $SU(2)$.

The holomorphic representation is a useful tool to analyze the spinfoam dynamics. By going beyond the previous single holomorphic vertex analysis \cite{Bianchi:2010mw}, we studied the partial amplitude for a general 2-complex and derived the constraints on the Ashtekar holonomies and on the conjugate gravitational fluxes which are enforced in the large area limit. In the case of a simplicial complex, we reproduced the semiclassical peaks in correspondence of Regge triangulations, well-known in other representations, and found a new strong constraint on the Regge deficit angles encoding curvature. In the class of limits considered here, only flat space-time geometries are allowed, up to a countable set of accidental curvatures.

Caution is needed if we want to draw conclusions about the flatness of the model from the previous analysis. The correct formulation of the semiclassical limit for the loop quantum gravity covariant dynamics might be more subtle than expected. We briefly discuss some possible scenarios.

It is possible that even though the amplitudes truncated to a finite 2-complex (the ones studied in this work) reproduce only the flat solution of Einstein equations, the continuum theory is well-defined and reproduces also the curved solutions. In this theoretical scenario the truncated theory can only allow for an infinitesimal deficit angle at the faces of the 2-complex in the semiclassical limit, which means that a finite curvature value can only be attained by adding up `many' infinitesimals, in agreement with a form of equivalence principle at the face. We leave this research direction as an open problem.

Within the truncated theory, a technical issue that we leave open is the question whether the constraints derived here, in particular the curvature constraint, continue to hold when we pass from the partial amplitudes to the full amplitudes obtained by integrating over the bulk coherent state labels. We hope to come back to this important point in the next future. It is possible that for some reason the curvature constraint does not hold for the full amplitude.

Another possible scenario is the one in which the semiclassical limit of the truncated theory is tightly related to flipped limit of small Barbero-Immirzi parameter. Indeed preliminary analyses \cite{Magliaro:2011zz,Magliaro:2011dz} suggest that the limit $\gamma\rightarrow 0$ could be used similarly to a semiclassical expansion, yielding non-trivial deficit angles. If this scenario turns out to be correct, light should be shed on the mechanism that explains the smallness of $\gamma$. Is it related to a renormalization of the Barbero-Immirzi parameter when we pass from the continuous amplitude to the amplitude truncated to finite graphs/2-complexes via coarse-graining? Is it related to the perturbative running \cite{Benedetti:2011nd}?
\section*{Acknowlegements}
Our thanks are due to Carlo Rovelli and Simone Speziale for hosting at the Centre de Physique ThŽorique de Luminy where most of the results were presented in detail on 1st October 2012. We also thank Eugenio Bianchi, Frank Hellmann and Daniele Oriti for useful comments.  
\bibliography{biblioholomorphic}
\appendix
\section{Proofs of section \ref{section: semiclassical vertex}}
\begin{proof}[Proof of proposition \ref{proposition: vertex}]
Let us write a $\SUtwoC$ coherent state label in the following way,
\begin{align}\label{BMPsplit}
H=he^{X/t^\beta}=hg(\hat X)e^{i|X|\tau_3/t^\beta}g(\hat X)^{-1},
\end{align}
where $g(\hat X)$ is a $SU(2)$ rotation that brings the reference unit 3-vector $(0,0,1)$ on the unit vector $\hat X$. There are infinitely many rotations with this property, so we assume this $U(1)$ ambiguity to be fixed by an arbitrary choice of a section $\hat X\rightarrow g(\hat X)$ of the Hopf bundle. Notice that however the equation \eqref{BMPsplit} depends only on a relative phase between $g(\hat X)$ and its inverse, thus it is independent of this choice.

As observed in \cite{Bianchi:2009ky}, the complexified diagonal $SU(2)$ representation matrix is dominated by the largest magnetic numbers for $t\rightarrow 0$, that is
\begin{align}
D^j_{mm'}(e^{i|X|\tau_3/t^\beta})=e^{j|X|/t^\beta}\delta_{mm'}\delta_{mj}(1+\mathcal O(t^\infty)\big),
\end{align}
where the notation $\mathcal O(t^\infty)$ means $\mathcal O(t^\alpha)$ for all $\alpha>0$. Substituting this in the expression \eqref{BMPsplit} for the coherent state label we get the approximation
\begin{align}\label{approximation of DjH}
D^j(H)=e^{j|X|/t^\beta}h\triangleright|j,\hat X\rangle\langle j,\hat X|\big(1+\mathcal O(t^\infty)\big),
\end{align}
where $|j,\hat X\rangle$ is a Bloch $SU(2)$ coherent state for the direction $\hat X$ in the spin-$j$ representation. The holomorphic wedge amplitude \eqref{holomorphic wedge} is easily computed,
\begin{align}
P_t(G,H)=\sum_j(2j+1)e^{-tj(j+1)}\text{Tr}[Y^\dagger D^{\gamma j,j}(G^{-1})YD^{j}(H)],
\end{align}
which using \eqref{approximation of DjH} can be approximated in the following way,
\begin{align}\label{holomorphic wedge approx}
P_t(G,H)\approx\sum_j(2j+1)e^{-tj(j+1)}e^{j|X|/t^\beta}\langle j,\hat X|Y^\dagger G^{-1}Yh|j,\hat X\rangle,
\end{align}
with relative error of order $\mathcal O(t^\infty)$. In order to study the peakedness properties of the vertex amplitude we need to write more explicitly the bracket in the previous formula \eqref{holomorphic wedge approx} for the wedge amplitude. It is sufficient to write explicitly the injection $Y$ and the inner product of the $SL(2,\mbb C)$ irreducible Hilbert space.

Thus we realize the Hilbert space $\mathcal H^{(j,\gamma j)}$ in the usual way as the space of homogeneous functions of two complex variables $z\in\mbb C^2$ of degree $(j,\gamma j)$, with the inner product
\begin{align}
(f,g):=\int_{\mbb C\mbb P^1}dz\,\omega(z)\overline{f(z)}g(z).
\end{align}
The measure $z\,\omega(z)$ is the standard invariant 2-form. The integration on the complex projective line $\mbb C\mbb P^1$ is well-defined because the 1-form we are integrating is scale-invariant. In other words, the variables $z=(z_1,z_2)\in\mbb C^2$ can be interpreted as homogeneous coordinates $[z_1:z_2]$ for $\mbb C\mbb P^1$. The standard norm of $z\in\mbb C^2$ is $|z|:=\sqrt{z^\dagger z}$. Also the Bloch coherent state in the fundamental representation belongs to $\mbb C^2$, thus we use the simpler notation $x$ for the vector $\ket{\frac{1}{2},\hat X}$. The image the homogeneous function realization of a Bloch SU(2) coherent state \cite{Barrett:2009mw} under the injection map $Y$ is the following,
\be\label{EPRLhom}
Y\ket{j,\hat X}\leadsto\sqrt\frac{2j+1}{\pi}\,|z|^{2i\gamma-2j-2}(z^\dagger x)^{2j}.
\ee
Using this expression in \eqref{holomorphic wedge approx}, we can write the holomorphic vertex amplitude \eqref{holomorphic vertex explicit} as
\begin{align}\label{vertex with action}
W_t(H_{vf})\approx\alpha_v\sum_{j_{vf}}\int_{SL(2,\mbb C)} dG_{ve}\int_{\mbb{CP}^1}dz_{vf} \mu_v\,e^{S_v},
\end{align}
where the function $S_v$, that we shall call vertex action, is the sum of three terms: $S_v=S^0_v+S'_v+S''_v$, where
\begin{align}\label{S0}
&S^0_v=\sum_{f\supset v}S^0_{vf}=2i\gamma  \sum_{f\supset v}j_{vf}\log Q^0_{vf},\\\label{S'}
&S'_v=\sum_{f\supset v}S'_{vf}=2\sum_{f\supset v}j_{vf}\log Q'_{vf},\\\label{S''}
&S''_v=\sum_{f\supset v}S''_{vf}=-t\sum_{f\supset v}\Big(j_{vf}-\frac{|X_{vf}|}{2t^{\beta+1}}\Big)^2.
\end{align}
The holomorphic vertex action $S_v$ is a function of the configuration $(j_{vf},z_{vf},G_{vf})$ we are integrating over.
The quantities inside the logarithm are easily computed using \eqref{EPRLhom}. These are
\begin{align}\label{Q}
&Q^0_{vf}=\frac{|G^\dagger_{v,s(vf)}z_{vf}|}{|G^\dagger_{v,t(vf)}z_{vf}|},\\\label{Q'}
&Q'_{vf}=\frac{x^\dagger_{vf}h^{\dagger}_{vf}G^\dagger_{v,t(vf)}z_{vf}}{|G^\dagger_{v,t(vf)}z_{vf}|}\frac{z^\dagger_{vf} G_{v,s(vf)} x_{vf}}{|G^\dagger_{v,s(vf)}z_{vf}|}.
\end{align}
The measure factor $\mu_v$ and the overall factor $\alpha_v$ are given by
\be
\mu_v=\prod_{f\supset v} (2j_{vf}+1)\frac{\sqrt\frac{2j_{vf}+1}{\pi}\,\omega(z_{vf})}{|G^\dagger_{v,t(vf)}z_{vf}|^2\,|G^\dagger_{v,s(vf)}z_{vf}|^2},
\ee
and
\begin{align}\label{overall alpha}
\alpha_v=\prod_{f\supset v}\exp{\frac{|X_{vf}|^2}{4t^{\beta+2}}}.
\end{align}
The overall factor $\alpha_v$ comes from completing the squares to form the Gaussian term $S''_v$ of the action. Notice that $S^0_v$ is purely imaginary while $S'_v$ is complex. Moreover we have 
\begin{align}
 |x^\dagger G z|\leq |x||G z|=|G z|,
\end{align}
for all normalized spinors $x$, which implies that $\text{Re }S_{vf}=\text{Re }S_{vf}'+S''_{vf}\leq 0$. In particular, the real part of the total vertex action $S_v$ is never positive. 

Heuristically, since the Gaussian term $S''_v$ peaks the sum on the large spins
\begin{align}
j_{vf}\simeq j^0_{vf}\simeq \frac{|X_{vf}|}{2t^{\beta+1}},
\end{align}
and the other two terms $S^0_v$ and $S'_v$ depend linearly on the spins, we realize that the contribution of a configuration $(j_{vf},z_{vf},G_{vf})$ to the integral is exponentially small in $t$ unless the real part of $S_v$ vanishes. Moreover, since the complex action $S_v$ is rapidly oscillating for $t\rightarrow 0$ in a neighborhood of $j^0_{vf}$, the holomorphic vertex amplitude is dominated by the stationary configurations of $S_v$. We call critical configuration a configuration $(j_{vf},z_{vf},G_{vf})$ for which the real part \emph{and} the gradient of $S$ vanish,
\begin{align}\label{ReS0}
\text{Re }S_v(j_{vf},z_{vf},G_{vf})&=0,\\\label{dS0}
\delta S_v(j_{vf},z_{vf},G_{vf})&=0,
\end{align}
where the variation is taken independently with respect to all the variables $j_{vf}$, $z_{vf}$ and $G_{ve}$. If there are no critical configurations, the integral is suppressed faster than any power of $t$. Considering the expression \eqref{overall alpha} of the overall factor $\alpha$, we find that the holomorphic vertex amplitude is non-suppressed in the sense of definition \ref{definition: suppression} if and only if the critical equations \eqref{ReS0}, \eqref{dS0} hold.

Notice that we are allowed to take variations with respect to the discrete variables $j_{vf}$ because they become quasi-continuous in the limit $t\rightarrow 0$, as for each wedge the density of spins $j_{vf}$ increases with respect to the interval $j^0_{vf}\pm \sqrt t$, where the summand is essentially different from zero. This heuristic argument can be made rigorous using the Euler-Maclaurin formula for the asymptotic approximation of the sum with the integral.

Let us start analyzing the solutions of the critical equations. The first equation \eqref{ReS0} is satisfied if and only if $\text{Re }S^0_{vf}=\text{Re }S'_{vf}=\text{Re }S''_{vf}=0$ for every face $f\supset v$. Using the previous explicit expressions \eqref{S0}, \eqref{S'}, \eqref{S''}, this happens if and only if
\begin{align}\label{re0Gauss}
&j_{vf}=\frac{|X_{vf}|}{2t^{\beta+1}},
\end{align}
and
\begin{align}
\frac{|x^\dagger_{vf}h^\dagger_{vf}G^\dagger_{v,t(vf)}z_{vf}|}{|G^\dagger_{v,t(vf)}z_{vf}|}&=1,\\
\frac{|z^\dagger_{vf} G_{v,s(vf)} x_{vf}|}{|G^\dagger_{v,s(vf)}z_{vf}|}&=1,
\end{align}
which may be rewritten as
\begin{align}\label{re01}
&\frac{G^\dagger_{v,t(vf)}z_{vf}}{|G^\dagger_{v,t(vf)}z_{vf}|}=e^{i\phi_{vf,t(vf)}}h_{vf}x_{vf},\\\label{re02}
&\frac{G^\dagger_{v,s(vf)}z_{vf}}{|G^\dagger_{v,s(vf)}z_{vf}|}=e^{i\phi_{vf,s(vf)}}x_{vf}.
\end{align}

Now let us analyze the the stationary phase equations \eqref{dS0} of the action $S_v$. For the variation with respect to the complex variables $z_{vf}$ we can consider independent variations with respect to $z_{vf}$ and $z^{\dagger}_{vf}$. If we do so we get easily
\begin{align}\label{varz}
&\delta_{z_{vf}}S_v=0 \quad\text{\emph{iff}} \quad2 \frac{x^\dagger_{vf}h^\dagger_{vf}G^\dagger_{v,t(vf)}}{x^\dagger_{vf}h^\dagger_{vf}G^\dagger_{v,t(vf)}z_{vf}}-\frac{z^\dagger_{vf}G_{v,t(vf)}G^\dagger_{v,t(vf)}}{|G^\dagger_{v,t(vf)}z_{vf}|^2}-\frac{z^\dagger_{vf}G_{v,s(vf)}G^\dagger_{v,s(vf)}}{|G^\dagger_{v,s(vf)}z_{vf}|^2}=0,
\end{align}
and
\begin{align}\label{varzdagger}
&\delta_{z^\dagger_{vf}}S_v=0 \quad\text{\emph{iff}} \quad2 \frac{G_{v,s(vf)}x_{vf}}{z_{vf}^{\dagger}G_{v,s(vf)}x_{vf}}-\frac{G_{v,t(vf)}G^\dagger_{v,t(vf)}z_{vf}}{|g^\dagger_{v,t(vf)}z_{vf}|^2}-\frac{G_{v,s(vf)}G^\dagger_{v,s(vf)}z_{vf}}{|g^\dagger_{v,s(vf)}z_{vf}|^2}=0.
\end{align}
Using \eqref{re01} and \eqref{re02}, the equations \eqref{varz} and \eqref{varzdagger} collapse into the single equation
\begin{align}\label{sfeom1}
\frac{G_{v,t(vf)} h_{vf}x_{vf}}{e^{-i\phi_{vf,t(vf)}}|G^\dagger_{v,t(vf)}z_{vf}|}=\frac{G_{v,s(vf)} x_{vf}}{e^{-i\phi_{vf,s(vf)}}|G^\dagger_{v,s(vf)}z_{vf}|},
\end{align}
where $\phi_{vfe}$ are phases. Notice also that eliminating $z_{vf}$ from \eqref{re01} and \eqref{re02} we get
\begin{align}\label{sfeom2}
\frac{G^*_{v,t(vf)} h_{vf}x_{vf}}{e^{-i\phi_{vf,t(vf)}}|G^\dagger_{v,s(vf)}z_{vf}|}=\frac{G^*_{v,s(vf)} x_{vf}}{e^{-i\phi_{vf,s(vf)}}|G^\dagger_{v,t(vf)}z_{vf}|}.
\end{align}

Let us study the variational equation with respect to the group variables $G_{ev}$, in the case $\epsilon_{vef}=1$ for all $f\supset e$. This yields the following,
\begin{align}
&\delta_{G_{ev}} S_v=0\quad\text{\emph{iff}}\quad\sum_{f\supset e} j_f\Big(2\frac{z^\dagger_{vf}G_{ve}Lx_{vf}}{z^\dagger_{vf}G_{ve}x_{vf}}-\frac{z^\dagger_{vf}G_{ve}Lg^\dagger_{ve}z_{vf}}{|g^\dagger_{ve}z_{vf}|^2}-\frac{z^\dagger_{vf}g_{ve}L^\dagger G^\dagger_{ve}z_{vf}}{|G^\dagger_{ve}z_{vf}|^2}\Big)=0,
\end{align}
where $L\in sl(2,\mbb C)$ is an arbitrary element of the Lorentz algebra. Using the conditions \eqref{re0Gauss}, \eqref{re01} and \eqref{re02} the previous equation becomes simply
\begin{align}\label{withL}
\sum_{f\supset e}|X_{vf}| \big(x^\dagger_{vf}Lx_{vf}-x^\dagger_{vf}L^\dagger x_{vf}\big)=0.
\end{align}
Now using $sl(2,\mbb C)\simeq su(2)\oplus i\,su(2)$ we write
\begin{align}
L=\vec\alpha\cdot\frac{\vec\sigma}{2}+i \vec\beta\cdot\frac{\vec\sigma}{2},
\end{align}
with $\vec\sigma$ the Pauli matrices. Since \eqref{withL} must hold for every variation, namely for all $\vec\alpha$ and $\vec\beta$ in $\mbb R^3$, the group variational equation is equivalent to the closure condition
\begin{align}
\sum_{f\supset e}X_{vf}=0,
\end{align}
where we have used $x^\dagger\vec\sigma x=\hat X$. A similar result holds for general orientations $\epsilon_{vef}$. In the general case the closure condition reads
\begin{align}\label{sfeom3}
\sum_{f_\text{in}\supset e}X_{vf}-\sum_{f_\text{out}\supset e}h_{vf}\triangleright X_{vf}=0,
\end{align}
which is the constraint \eqref{closure} of proposition \ref{proposition: vertex}.

The equations \eqref{sfeom1} and \eqref{sfeom2} can be easily casted into the transport equation
\begin{align}\label{transport equation proof}
U(h^{-1}_{vf}G_{t(vf),v}G_{v,s(vf)})=e^{\xi_{vf}^* B(X_{vf})+\xi_{vf} *B(X_{vf})},
\end{align}
where
\begin{align}
\xi_{vf}&:=2\log Q^0_{vf},\\
\xi^*_{vf}&:=-2i\log Q'_{vf}=2(\phi_{vf,t(vf)}-\phi_{vf,s(vf)}).
\end{align}
Notice that $\xi_{vf}$ is a real number (the rapidity of a boost), and $\xi^*_{vf}$ is a phase (a rotation angle). 

Finally, the last stationary phase equation comes from the spin variations and reads
\begin{align}\label{spin variational}
\delta_{j_{vf}}S_v=0 \quad\text{\emph{iff}} \quad i\gamma\xi_{vf} + i\xi^*_{vf}\mod 4\pi i=t\Big(j_{vf}-\frac{|X_{vf}|}{2t^2}\Big),
\end{align}
where the right-hand side vanishes by \eqref{re0Gauss}. Notice that since the phase $\xi^*_{vf}$ is defined by the spinorial equations \eqref{re01}, \eqref{re02}, the correct periodicity is $4\pi$. Thus the spin variational equation becomes simply
\begin{align}
\gamma\xi_{vf} + \xi^*_{vf}=0\mod 4\pi,
\end{align}
which together with \eqref{transport equation proof} gives the constraint \eqref{wedge relation} of proposition \ref{proposition: vertex}, and concludes the proof.
\end{proof}


\section{Proofs of section \ref{section: semiclassical face}}
\begin{proof}[Proof of lemma \ref{proposition: face amplitude}] Consider a face $f$ of the 2-complex. Notice that in order to prove the result we can study equivalently the complex conjugate of the face amplitude. Using the approximation \eqref{approximation of DjH}, the complex conjugate of the face amplitude for the face can be rewritten disregarding an exponentially small relative contribution as
\begin{align}
K_{V(f)t}(\prod_{v\subset f}H_{fv})\approx \alpha_f\sum_{j_f} \mu_f\,e^{S_f},
\end{align}
with face action $S_f=S'_f+S''_f$ defined by
\begin{align}
&S'_f=\sum_{e\subset f}S'_{fe}=2 j_f \sum_{e\subset f} \log Q_{fe},\\
&S''_f=-t\sum_{v\subset f}\Big(j_f-\frac{|X_{vf}|}{2t^{\beta+1}}\Big)^2,
\end{align}
where the measure is simply $\mu_f=2j_f+1$, the overall factor $\alpha_f$ is 
\begin{align}
\alpha_f=\prod_{v\subset f}\exp{\frac{|X_{vf}|^2}{4t^{\beta+2}}},
\end{align}
and the argument of the logarithm is given by
\begin{align}\label{Qfe}
Q_{fe}=x_{t(e)f}^\dagger h_{s(e)f}x_{s(e)f}.
\end{align}
Remember that the normalized spinor $x_{vf}$ is defined as the Bloch coherent state $|1/2,\hat X_{vf}\rangle$ in the fundamental representation, up to a phase ambiguity that is canceled out by the conjugate state $x^\dagger_{vf}$. Thus \eqref{Qfe} is well-defined. Rather, the relative face between the source coherent state, transported by $h_{vf}$ to the target point, and the target state is defined unambiguously by the Ashtekar holonomy $h_{vf}$ itself.

As we did for the holomorphic vertex amplitude, the non-suppressed configurations must satisfy the critical equations. Now there is a single configuration variable, the spin $j$, we are integrating over. Moreover, the spin is quasi-continuous in the semiclassical limit $t\rightarrow 0$ and we can approximate the sum with an integral. Similarly to the vertex amplitude, the real part of the action is non-positive and vanishes if and only if\begin{align}\label{face first}
h_{s(e)f}x_{s(e)f}=e^{i\phi_{fe}/2} x_{t(e)f},
\end{align}
for all edges $e\subset f$, and
\begin{align}\label{face second}
j_f=\frac{|X_{vf}|}{2t^{\beta+1}},
\end{align}
for all vertices $v\subset f$. Notice that the first equation \eqref{face first} is a gluing constraint for directions of the fluxes at the interface of two vertices. The second equation \eqref{face second} implies in particular the area matching constraint, namely the moduli of the fluxes labeled by the same face must be equal. The previous two equations are thus equivalent to the flux glueing constraint, equation \eqref{glueing constraint} of the proposition.

The second critical equation is the spin variational equation, which yields
\begin{align}\label{face spin var}
\delta_j S_f=0 \quad\text{\emph{iff}} \quad \sum_{e\subset f}\phi_{fe}=0\mod 4\pi,
\end{align}
where the phases are defined in \eqref{face first}. Since equation \eqref{face first} implies that
\begin{align}
(\prod_{v\subset f}h_{vf})x_{vf}=e^{i\sum_{e\subset f}\phi_{fe}/2}x_{vf},
\end{align}
where the cyclic product is ordered according to the orientation of the face and the choice of the first element of the product, here labeled with a generic $vf$, is arbitrary, the variational equation \eqref{face spin var} implies that
\begin{align}
\sum_{e\subset f}\phi_{fe}=0\mod 4\pi,
\end{align}
namely
\begin{align}
\prod_{v\subset f}h_{vf}=\mbb 1,
\end{align}
that is the Ashtekar loop holonomy is the identity. 
\end{proof}

\section{Proofs of section \ref{section: 4-simplex}}
\begin{proof}[Proof of lemma \ref{lemma: classification vertex solutions}]
Consider boundary holonomy-flux labels $H_{vf}$ for a vertex graph $\Gamma_v$, and $SL(2,\mbb C)$ elements $G_{ve}$, one for each edge $e\subset v$, such that they determine a non-degenerate bivector geometry via \eqref{transport constraint}. Consider the family of $SL(2,\mbb C)$ elements obtained from $G_{ve}$ by rigid $SL(2,\mbb C)$ transformations, parity, and spin lift symmetry. Each set of elements $G'_{ve}$ in this family determines a non-degenerate bivector geometry. This can be easily checked by direct inspection and proves the existence part of the lemma.

To prove the uniqueness part, consider a bivector geometry, not necessarily non-degenerate, determined by the same holonomy-flux labels $H_{vf}$, and some other elements $G'_{ve}$. Now we use the following classification of the boundary data (see \cite{Barrett:2009mw}). The holonomy-flux labels $H_{vf}$ determine in particular a 3-dimensional Riemannian Regge triangulation of the 3-sphere, with unique edge lengths. Since the edge lengths fix uniquely the metric of a 4-simplex, the triangulation can be: 1) the boundary of a Lorentzian 4-simplex, 2) the boundary of a Euclidean 4-simplex, 3) the boundary of a degenerate 4-simplex. These cases are mutually exclusive for fixed edge lengths. However, the existence of the elements $G_{ve}$ which determine a non-degenerate Lorentzian 4-simplex, automatically excludes the cases 2) and 3). Moreover, for fixed edge lengths, the Lorentzian 4-simplex is determined up to Poincar\'e transformations. The spin-lift symmetry is irrelevant for the classification. Taking into account the relation between bivector geometry and geometric 4-simplices, the lemma is proven observing that the rigid $SL(2,\mbb C)$ symmetry and the parity symmetry of the elements $G'_{ve}$ encompass all the possible Poincar\'e transformations, since the space-time inversion and the translations are not registered by the bivector geometry.
\end{proof}

\section{Proofs of section \ref{section: Ashtekar twists}}
\begin{proof}[Proof of lemma \ref{lemma: natural decomposition of GG}]
Given a reconstructed 4-simplex $(\sigma_{v},\mu_v)$, consider a wedge $vf$ and call $e$, $e'$ the source and target edges respectively. We shall use the notation $e'e$ instead of $vf$ to denote this wedge. Notice also that $e'e$ identifies the triangle shared by the tetrahedra $e^*$ and $e'^*$.

Suppose both external normals $N_{e}$ and $N_{e'}$ are future-pointing, or past-pointing. Use the space-time inversion ambiguity to fix them to be both future-pointing. Now we have fixed a unique 4-simplex in $\Mink$, up to irrelevant translations. Applying $G_{ev}$ and $G_{e'v}$ to the tetrahedra $e$ and $e'$ respectively, we bring them back to the 3-plane orthogonal to $\mathcal N$. Notice that since we act with a $SL(2,\mbb C)$ transformation, i.e. a proper orthochronous Lorentz transformation, the future-pointing normals are both sent to $\mathcal N=(1,0,0,0)$ (and not to minus $\mathcal N$). Clearly, there is a unique diagonal $SU(2)$ rotation $g_{e'e}$ that brings the triangle $f$ of the transformed source tetrahedron onto the triangle $f$ of the transformed target tetrahedron. Their 3-dimensional external normals are automatically sent by $g_{e'e}$ in the anti-parallel configuration, for we used proper orthochronous elements to put the tetrahedra in the time-gauge. Thus the cycle $G_{ve'}g_{e'e}G_{ev}$ is a Lorentz transformation that preserves the triangle $f$ of the 4-simplex, and sends $N_e$ to $N_{e'}$. Hence we have
\begin{align}
U(G_{ve'}g_{e'e}G_{ev})=e^{|\Theta|N_{e}\wedge N_{e'}},
\end{align}
where we put the modulus to emphasize the positivity of the boost parameter. Notice that this result is independent of the inversion ambiguity.

Now consider the second case in which one 4-dimensional normal is future-pointing, and the other is past-pointing. Use the inversion ambiguity to fix the source one to be past-pointing. Doing a similar analysis, in this case we have that the cycle $G_{ve'}g_{e'e}G_{ev}$ preserves the triangle $f$ as before, but sends \emph{minus} $N_e$ to $N_{e'}$. So we must have
\begin{align}
U(G_{ve'}g_{e'e}G_{ev})=e^{|\Theta|(-N_{e})\wedge N_{e'}}=e^{\Theta N_{e}\wedge N_{e'}},
\end{align}
with $\Theta$ negative. Thus this equation is the same as for the first case, $\Theta$ being the Lorentzian dihedral angle at the triangle $e'e$. Thus we have shown the equation \eqref{natural 2}.

Multiplying this equation by $G_{ve}$ on the right and by $G_{ev}$ on the left, and taking the inverse, we arrive to
\begin{align}\label{GG=gexpGNN}
U(g^{-1}_{e'e}G_{e'v}G_{ve})=e^{-\Theta G_{ev}(N_{e}\wedge N_{e'})}.
\end{align}
By the Barrett-Crane reconstruction theorem, 
\begin{align}
N_{e}\wedge N_{e'}=*_{\mu_{v}}*_{\mu_{v}}N_{e}\wedge N_{e'}=*_{\mu_{v}}G_{ve}B(X_{e'e}),
\end{align}
and substituting this in \eqref{GG=gexpGNN} we get the alternative form \eqref{natural 1} of the decomposition, which concludes the proof.
\end{proof}

\section{Proofs of section \ref{section: 4-simplex glueing}}
\begin{proof}[Proof of proposition \ref{proposition: 4-simplex glueing}]
All the analysis in this proof is localized at a edge $e$ of the 2-complex, which is the interface of two vertices $v$ and $v'$.

Using the Barrett-Crane reconstruction and the hypotheses of the proposition we can associate two 4-simplices $(\sigma_v,\mu_v)$ and $(\sigma_{v'},\mu_{v'})$ to the two sets of holonomy-flux labels. For simplicity, we are interested to the time-gauge the two tetrahedra labeled by the same $e$, one in the 4-simplex $\sigma_v$ and one in the 4-simplex $\sigma_{v'}$. We first do this transformation at the bivectors level.

So we apply the inverse rigid transformations $G_{ve}^{-1}$ and $G_{v'e}^{-1}$ to all the bivectors $B_{f}(v)$ ($f\supset v$) and $B_f(v')$ ($f\supset v'$) respectively. Using the glueing equation \eqref{glueing constraint} found in the analysis of the anti-holomorphic face amplitude, the eight transformed-back bivectors with $f\supset e$, four in $v$ and four in $v'$, must be  identified pairwise with
\be
\frac12 a_f *(1,\vec n_{ef})\wedge(1,-\vec n_{ef}),
\ee 
where $a_f:=|X_{vf}|$, and we have defined the edge-face $\mbb R^3$ unit vectors
\begin{align}
\vec n_{ef}:=
\begin{cases}
\hat X_{v'f}=h_{vf}\hat X_{vf},\quad \epsilon_{evf}=+1\\
\hat X_{vf}=h_{v'f}\hat X_{v'f},\quad \epsilon_{evf}=-1
\end{cases}
\end{align}
which close to zero according to the equation
\begin{align}
\sum_{f\supset e}\epsilon_{vef}a_f\vec n_{ef}=0.
\end{align}

Thus besides $(\sigma_v,\mu_v)$ and $(\sigma_{v'},\mu_{v'})$, we are interested in the oriented 4-simplices $(\tilde\sigma_v,\mu_v)$ and $(\tilde\sigma_{v'},\mu_{v'})$ obtained by the Barrett-Crane reconstruction from the $20=10+10$ bivectors $G_{ve}^{-1}B_{f}(v)$ and $G_{v'e}^{-1}B_{f}(v')$. The tilde serves to denote the 4-simplices which are Lorentz-transformed-back, where the tetrahedron $e$ is in the time-gauge. Notice that a $SL(2,\mbb C)$ transformation cannot change the reconstructed orientations.

The external unit normals $N_{e}(v)$ and $N_{e}(v')$ of the 4-simplices $\tilde\sigma_v$ and $\tilde\sigma_{v'}$ are given by $(\pm 1,0,0,0)$ depending on which inversion-related 4-simplices are chosen in the reconstruction. Fix the inversion ambiguity using a choice parametrized with $p_v=\pm1$:
\be\label{realization}
N_{e}(v)=p_v(1,0,0,0),\quad N_{e}(v')=p_{v'}(1,0,0,0).
\ee
The two tetrahedra $e$ in $\tilde\sigma_v$ and $\tilde\sigma_{v'}$ are in the subspace orthogonal to $(1,0,0,0)$, that can be identified with $\mbb R^3$ by erasing the time coordinate. Using this identification, we have that in the case the reconstructed orientations are $\mu_v=\mu_{v'}=1$, the four unit vectors $p_v\epsilon_{vef}\vec n_{ef}$ are the external normals of the tetrahedron $e$ in $\tilde\sigma_v$, and the unit vectors $p_{v'}\epsilon_{v'ef}\vec n_{ef}$ are the external normals of the second tetrahedron $e$ in $\tilde\sigma_{v'}$. Notice that $\epsilon_{vef}=-\epsilon_{v'ef}$. Thus we have two cases. If the inversion choice is $p_v=-p_{v'}$, then the two tetrahedra are the same up to translations and the 4-dimensional external normals are the same. Whereas if $p_v=p_{v'}$ then the two tetrahedra are related by spatial inversion $\vec x\rightarrow-\vec x$, and the 4-dimensional normals are related by time inversion, thus the tetrahedra \emph{and} their normals are related by a space-time inversion. 

Thus we have shown that given two adjacent reconstructed oriented 4-simplices $(\sigma_v,+)$ and $(\sigma_{v'},+)$, there is a $SO(1,3)$ transformation given by $\mu_e \text{Ad}(G_{v'e}G_{ev})$ that together with a translation glues the corresponding tetrahedra and brings their external normals to the antiparallel configuration. The inversion sign $\mu_e:=p_v p_{v'}$ depends on the way we choose to fix the inversion ambiguity. By \eqref{tetrahedraglueing} and \eqref{normalsglueing} this transformation is the Levi-Civita holonomy, which proves the first case in \eqref{eq: glueing theorem}.

Now let us analyze the other possible orientations. The previous argument can be repeated for the case $\mu_{v}=\mu_{v'}=-1$. So let us analyze the remaining two cases $\mu_v\neq\mu_{v'}$. Now if we choose $p_v=-p_{v'}$, we have that the external 3-normals of the two tetrahedra are opposite, and the external 4-normals are opposite. Thus the parity transformation $P:\vec x\rightarrow-\vec x$, namely the spatial inversion, sends the tetrahedron into the other tetrahedron leaving the 4-normals antiparallel. If we choose another realization of the inversion ambiguity such that $p_v=p_{v'}$, we have that the time inversion $T: t\rightarrow -t$, which is also the composition of parity with an inversion, makes the same job. This concludes the proof of the second case in \eqref{eq: glueing theorem}.
\end{proof}

\end{document}